\documentclass[10pt,conference,a4paper]{IEEEtran}

\usepackage[all]{xy}
\usepackage{xcolor}
\usepackage{mathtools}
\usepackage{proof}
\usepackage{cleveref}
\usepackage{tikzit}
\usepackage{wrapfig}
\usepackage{url}
\usetikzlibrary{shapes.gates.ee,shapes.gates.ee.IEC,shapes.geometric,shapes.gates.logic.US}
\usetikzlibrary{calc,patterns,decorations.pathreplacing}

\usepackage{latexsym}
\usepackage{amssymb}
\usepackage{amsmath}
\usepackage{xparse}
\usepackage[amsmath,thmmarks]{ntheorem}
\usepackage{enumerate}

\newcommand{\Bf}[1]{{\bf #1}}

\newcommand{\Rm}[1]{{\rm #1}}

\newcommand{\mc}{\mathcal}

\newcommand{\fa}[1]{\forall{#1}~.~}

\newcommand{\dom}{\Rm{dom}}
\newcommand{\cod}{\Rm{cod}}

\newcommand{\Set}{\Bf{Set}}

\newcommand{\arrow}{\rightarrow}

\newcommand{\bul}{\bullet}

\newcommand{\id}{\Rm{id}}

\newcommand{\ox}{\otimes}

\newcommand{\CC}{\mathbb{C}}

\newcommand{\NN}{\mathbb{N}}

\newcommand{\RR}{\mathbb{R}}

\newcommand{\rue}{\ar@{=}[u]}
\newcommand{\rueh}[1]{\ar@{=}[u]^-{#1}}
\newcommand{\ruem}[1]{\ar@{=}[u]_-{#1}}
\newcommand{\ruen}[1]{\ar@{=}[u]|-{#1}}

\newcommand{\ruue}{\ar@{=}[uu]}
\newcommand{\ruueh}[1]{\ar@{=}[uu]^-{#1}}
\newcommand{\ruuem}[1]{\ar@{=}[uu]_-{#1}}
\newcommand{\ruuen}[1]{\ar@{=}[uu]|-{#1}}

\newcommand{\ruuue}{\ar@{=}[uuu]}
\newcommand{\ruuueh}[1]{\ar@{=}[uuu]^-{#1}}
\newcommand{\ruuuem}[1]{\ar@{=}[uuu]_-{#1}}
\newcommand{\ruuuen}[1]{\ar@{=}[uuu]|-{#1}}

\newcommand{\rdh}[1]{\ar[d]^-{#1}}
\newcommand{\rdm}[1]{\ar[d]_-{#1}}

\newcommand{\rde}{\ar@{=}[d]}
\newcommand{\rdeh}[1]{\ar@{=}[d]^-{#1}}
\newcommand{\rdem}[1]{\ar@{=}[d]_-{#1}}
\newcommand{\rden}[1]{\ar@{=}[d]|-{#1}}

\newcommand{\rdde}{\ar@{=}[dd]}
\newcommand{\rddeh}[1]{\ar@{=}[dd]^-{#1}}
\newcommand{\rddem}[1]{\ar@{=}[dd]_-{#1}}
\newcommand{\rdden}[1]{\ar@{=}[dd]|-{#1}}

\newcommand{\rddde}{\ar@{=}[ddd]}
\newcommand{\rdddeh}[1]{\ar@{=}[ddd]^-{#1}}
\newcommand{\rdddem}[1]{\ar@{=}[ddd]_-{#1}}
\newcommand{\rddden}[1]{\ar@{=}[ddd]|-{#1}}

\newcommand{\rrh}[1]{\ar[r]^-{#1}}
\newcommand{\rrm}[1]{\ar[r]_-{#1}}

\newcommand{\rre}{\ar@{=}[r]}
\newcommand{\rreh}[1]{\ar@{=}[r]^-{#1}}
\newcommand{\rrem}[1]{\ar@{=}[r]_-{#1}}
\newcommand{\rren}[1]{\ar@{=}[r]|-{#1}}

\newcommand{\rrue}{\ar@{=}[ru]}
\newcommand{\rrueh}[1]{\ar@{=}[ru]^-{#1}}
\newcommand{\rruem}[1]{\ar@{=}[ru]_-{#1}}
\newcommand{\rruen}[1]{\ar@{=}[ru]|-{#1}}

\newcommand{\rruue}{\ar@{=}[ruu]}
\newcommand{\rruueh}[1]{\ar@{=}[ruu]^-{#1}}
\newcommand{\rruuem}[1]{\ar@{=}[ruu]_-{#1}}
\newcommand{\rruuen}[1]{\ar@{=}[ruu]|-{#1}}

\newcommand{\rruuue}{\ar@{=}[ruuu]}
\newcommand{\rruuueh}[1]{\ar@{=}[ruuu]^-{#1}}
\newcommand{\rruuuem}[1]{\ar@{=}[ruuu]_-{#1}}
\newcommand{\rruuuen}[1]{\ar@{=}[ruuu]|-{#1}}

\newcommand{\rrde}{\ar@{=}[rd]}
\newcommand{\rrdeh}[1]{\ar@{=}[rd]^-{#1}}
\newcommand{\rrdem}[1]{\ar@{=}[rd]_-{#1}}
\newcommand{\rrden}[1]{\ar@{=}[rd]|-{#1}}

\newcommand{\rrdde}{\ar@{=}[rdd]}
\newcommand{\rrddeh}[1]{\ar@{=}[rdd]^-{#1}}
\newcommand{\rrddem}[1]{\ar@{=}[rdd]_-{#1}}
\newcommand{\rrdden}[1]{\ar@{=}[rdd]|-{#1}}

\newcommand{\rrddde}{\ar@{=}[rddd]}
\newcommand{\rrdddeh}[1]{\ar@{=}[rddd]^-{#1}}
\newcommand{\rrdddem}[1]{\ar@{=}[rddd]_-{#1}}
\newcommand{\rrddden}[1]{\ar@{=}[rddd]|-{#1}}

\newcommand{\rrre}{\ar@{=}[rr]}
\newcommand{\rrreh}[1]{\ar@{=}[rr]^-{#1}}
\newcommand{\rrrem}[1]{\ar@{=}[rr]_-{#1}}
\newcommand{\rrren}[1]{\ar@{=}[rr]|-{#1}}

\newcommand{\rrrue}{\ar@{=}[rru]}
\newcommand{\rrrueh}[1]{\ar@{=}[rru]^-{#1}}
\newcommand{\rrruem}[1]{\ar@{=}[rru]_-{#1}}
\newcommand{\rrruen}[1]{\ar@{=}[rru]|-{#1}}

\newcommand{\rrruue}{\ar@{=}[rruu]}
\newcommand{\rrruueh}[1]{\ar@{=}[rruu]^-{#1}}
\newcommand{\rrruuem}[1]{\ar@{=}[rruu]_-{#1}}
\newcommand{\rrruuen}[1]{\ar@{=}[rruu]|-{#1}}

\newcommand{\rrruuue}{\ar@{=}[rruuu]}
\newcommand{\rrruuueh}[1]{\ar@{=}[rruuu]^-{#1}}
\newcommand{\rrruuuem}[1]{\ar@{=}[rruuu]_-{#1}}
\newcommand{\rrruuuen}[1]{\ar@{=}[rruuu]|-{#1}}

\newcommand{\rrrde}{\ar@{=}[rrd]}
\newcommand{\rrrdeh}[1]{\ar@{=}[rrd]^-{#1}}
\newcommand{\rrrdem}[1]{\ar@{=}[rrd]_-{#1}}
\newcommand{\rrrden}[1]{\ar@{=}[rrd]|-{#1}}

\newcommand{\rrrdde}{\ar@{=}[rrdd]}
\newcommand{\rrrddeh}[1]{\ar@{=}[rrdd]^-{#1}}
\newcommand{\rrrddem}[1]{\ar@{=}[rrdd]_-{#1}}
\newcommand{\rrrdden}[1]{\ar@{=}[rrdd]|-{#1}}

\newcommand{\rrrddde}{\ar@{=}[rrddd]}
\newcommand{\rrrdddeh}[1]{\ar@{=}[rrddd]^-{#1}}
\newcommand{\rrrdddem}[1]{\ar@{=}[rrddd]_-{#1}}
\newcommand{\rrrddden}[1]{\ar@{=}[rrddd]|-{#1}}

\newcommand{\rrrre}{\ar@{=}[rrr]}
\newcommand{\rrrreh}[1]{\ar@{=}[rrr]^-{#1}}
\newcommand{\rrrrem}[1]{\ar@{=}[rrr]_-{#1}}
\newcommand{\rrrren}[1]{\ar@{=}[rrr]|-{#1}}

\newcommand{\rrrrue}{\ar@{=}[rrru]}
\newcommand{\rrrrueh}[1]{\ar@{=}[rrru]^-{#1}}
\newcommand{\rrrruem}[1]{\ar@{=}[rrru]_-{#1}}
\newcommand{\rrrruen}[1]{\ar@{=}[rrru]|-{#1}}

\newcommand{\rrrruue}{\ar@{=}[rrruu]}
\newcommand{\rrrruueh}[1]{\ar@{=}[rrruu]^-{#1}}
\newcommand{\rrrruuem}[1]{\ar@{=}[rrruu]_-{#1}}
\newcommand{\rrrruuen}[1]{\ar@{=}[rrruu]|-{#1}}

\newcommand{\rrrruuue}{\ar@{=}[rrruuu]}
\newcommand{\rrrruuueh}[1]{\ar@{=}[rrruuu]^-{#1}}
\newcommand{\rrrruuuem}[1]{\ar@{=}[rrruuu]_-{#1}}
\newcommand{\rrrruuuen}[1]{\ar@{=}[rrruuu]|-{#1}}

\newcommand{\rrrrde}{\ar@{=}[rrrd]}
\newcommand{\rrrrdeh}[1]{\ar@{=}[rrrd]^-{#1}}
\newcommand{\rrrrdem}[1]{\ar@{=}[rrrd]_-{#1}}
\newcommand{\rrrrden}[1]{\ar@{=}[rrrd]|-{#1}}

\newcommand{\rrrrdde}{\ar@{=}[rrrdd]}
\newcommand{\rrrrddeh}[1]{\ar@{=}[rrrdd]^-{#1}}
\newcommand{\rrrrddem}[1]{\ar@{=}[rrrdd]_-{#1}}
\newcommand{\rrrrdden}[1]{\ar@{=}[rrrdd]|-{#1}}

\newcommand{\rrrrddde}{\ar@{=}[rrrddd]}
\newcommand{\rrrrdddeh}[1]{\ar@{=}[rrrddd]^-{#1}}
\newcommand{\rrrrdddem}[1]{\ar@{=}[rrrddd]_-{#1}}
\newcommand{\rrrrddden}[1]{\ar@{=}[rrrddd]|-{#1}}

\newcommand{\rle}{\ar@{=}[l]}
\newcommand{\rleh}[1]{\ar@{=}[l]^-{#1}}
\newcommand{\rlem}[1]{\ar@{=}[l]_-{#1}}
\newcommand{\rlen}[1]{\ar@{=}[l]|-{#1}}

\newcommand{\rlue}{\ar@{=}[lu]}
\newcommand{\rlueh}[1]{\ar@{=}[lu]^-{#1}}
\newcommand{\rluem}[1]{\ar@{=}[lu]_-{#1}}
\newcommand{\rluen}[1]{\ar@{=}[lu]|-{#1}}

\newcommand{\rluue}{\ar@{=}[luu]}
\newcommand{\rluueh}[1]{\ar@{=}[luu]^-{#1}}
\newcommand{\rluuem}[1]{\ar@{=}[luu]_-{#1}}
\newcommand{\rluuen}[1]{\ar@{=}[luu]|-{#1}}

\newcommand{\rluuue}{\ar@{=}[luuu]}
\newcommand{\rluuueh}[1]{\ar@{=}[luuu]^-{#1}}
\newcommand{\rluuuem}[1]{\ar@{=}[luuu]_-{#1}}
\newcommand{\rluuuen}[1]{\ar@{=}[luuu]|-{#1}}

\newcommand{\rlde}{\ar@{=}[ld]}
\newcommand{\rldeh}[1]{\ar@{=}[ld]^-{#1}}
\newcommand{\rldem}[1]{\ar@{=}[ld]_-{#1}}
\newcommand{\rlden}[1]{\ar@{=}[ld]|-{#1}}

\newcommand{\rldde}{\ar@{=}[ldd]}
\newcommand{\rlddeh}[1]{\ar@{=}[ldd]^-{#1}}
\newcommand{\rlddem}[1]{\ar@{=}[ldd]_-{#1}}
\newcommand{\rldden}[1]{\ar@{=}[ldd]|-{#1}}

\newcommand{\rlddde}{\ar@{=}[lddd]}
\newcommand{\rldddeh}[1]{\ar@{=}[lddd]^-{#1}}
\newcommand{\rldddem}[1]{\ar@{=}[lddd]_-{#1}}
\newcommand{\rlddden}[1]{\ar@{=}[lddd]|-{#1}}

\newcommand{\rlle}{\ar@{=}[ll]}
\newcommand{\rlleh}[1]{\ar@{=}[ll]^-{#1}}
\newcommand{\rllem}[1]{\ar@{=}[ll]_-{#1}}
\newcommand{\rllen}[1]{\ar@{=}[ll]|-{#1}}

\newcommand{\rllue}{\ar@{=}[llu]}
\newcommand{\rllueh}[1]{\ar@{=}[llu]^-{#1}}
\newcommand{\rlluem}[1]{\ar@{=}[llu]_-{#1}}
\newcommand{\rlluen}[1]{\ar@{=}[llu]|-{#1}}

\newcommand{\rlluue}{\ar@{=}[lluu]}
\newcommand{\rlluueh}[1]{\ar@{=}[lluu]^-{#1}}
\newcommand{\rlluuem}[1]{\ar@{=}[lluu]_-{#1}}
\newcommand{\rlluuen}[1]{\ar@{=}[lluu]|-{#1}}

\newcommand{\rlluuue}{\ar@{=}[lluuu]}
\newcommand{\rlluuueh}[1]{\ar@{=}[lluuu]^-{#1}}
\newcommand{\rlluuuem}[1]{\ar@{=}[lluuu]_-{#1}}
\newcommand{\rlluuuen}[1]{\ar@{=}[lluuu]|-{#1}}

\newcommand{\rllde}{\ar@{=}[lld]}
\newcommand{\rlldeh}[1]{\ar@{=}[lld]^-{#1}}
\newcommand{\rlldem}[1]{\ar@{=}[lld]_-{#1}}
\newcommand{\rllden}[1]{\ar@{=}[lld]|-{#1}}

\newcommand{\rlldde}{\ar@{=}[lldd]}
\newcommand{\rllddeh}[1]{\ar@{=}[lldd]^-{#1}}
\newcommand{\rllddem}[1]{\ar@{=}[lldd]_-{#1}}
\newcommand{\rlldden}[1]{\ar@{=}[lldd]|-{#1}}

\newcommand{\rllddde}{\ar@{=}[llddd]}
\newcommand{\rlldddeh}[1]{\ar@{=}[llddd]^-{#1}}
\newcommand{\rlldddem}[1]{\ar@{=}[llddd]_-{#1}}
\newcommand{\rllddden}[1]{\ar@{=}[llddd]|-{#1}}

\newcommand{\rllle}{\ar@{=}[lll]}
\newcommand{\rllleh}[1]{\ar@{=}[lll]^-{#1}}
\newcommand{\rlllem}[1]{\ar@{=}[lll]_-{#1}}
\newcommand{\rlllen}[1]{\ar@{=}[lll]|-{#1}}

\newcommand{\rlllue}{\ar@{=}[lllu]}
\newcommand{\rlllueh}[1]{\ar@{=}[lllu]^-{#1}}
\newcommand{\rllluem}[1]{\ar@{=}[lllu]_-{#1}}
\newcommand{\rllluen}[1]{\ar@{=}[lllu]|-{#1}}

\newcommand{\rllluue}{\ar@{=}[llluu]}
\newcommand{\rllluueh}[1]{\ar@{=}[llluu]^-{#1}}
\newcommand{\rllluuem}[1]{\ar@{=}[llluu]_-{#1}}
\newcommand{\rllluuen}[1]{\ar@{=}[llluu]|-{#1}}

\newcommand{\rllluuue}{\ar@{=}[llluuu]}
\newcommand{\rllluuueh}[1]{\ar@{=}[llluuu]^-{#1}}
\newcommand{\rllluuuem}[1]{\ar@{=}[llluuu]_-{#1}}
\newcommand{\rllluuuen}[1]{\ar@{=}[llluuu]|-{#1}}

\newcommand{\rlllde}{\ar@{=}[llld]}
\newcommand{\rllldeh}[1]{\ar@{=}[llld]^-{#1}}
\newcommand{\rllldem}[1]{\ar@{=}[llld]_-{#1}}
\newcommand{\rlllden}[1]{\ar@{=}[llld]|-{#1}}

\newcommand{\rllldde}{\ar@{=}[llldd]}
\newcommand{\rlllddeh}[1]{\ar@{=}[llldd]^-{#1}}
\newcommand{\rlllddem}[1]{\ar@{=}[llldd]_-{#1}}
\newcommand{\rllldden}[1]{\ar@{=}[llldd]|-{#1}}

\newcommand{\rlllddde}{\ar@{=}[lllddd]}
\newcommand{\rllldddeh}[1]{\ar@{=}[lllddd]^-{#1}}
\newcommand{\rllldddem}[1]{\ar@{=}[lllddd]_-{#1}}
\newcommand{\rlllddden}[1]{\ar@{=}[lllddd]|-{#1}}

\newcommand{\rmid}[2]{\ar@{}[#1]|-{#2}}

\newcommand{\adj}[3]{
  \ar@<+.3pc>@{->}[#1]^-{#2} 
  \ar@<-.3pc>@{<-}[#1]_-{#3} 
  \ar@{}[#1]|-{\dashv} 
}

\newtheorem{theorem}{Theorem}
\newtheorem{proposition}[theorem]{Proposition}
\newtheorem{lemma}[theorem]{Lemma}
\newtheorem{corollary}[theorem]{Corollary}

\newtheorem{example}[theorem]{Example}

\newtheorem{definition}[theorem]{Definition}

\theoremstyle{nonumberplain}
\theorembodyfont{\normalfont}
\theoremsymbol{\ensuremath{\square}}
\newtheorem{proof}{Proof}

\tikzstyle{qstate}=[fill={rgb,255: red,255; green,128; blue,0}, draw=black, shape=rectangle, tikzit fill={rgb,255: red,255; green,128; blue,0}]
\tikzstyle{tallbox}=[fill=white, draw=black, shape=rectangle, minimum height=1.25cm, minimum width=0.75cm]
\tikzstyle{box}=[fill=white, draw=black, shape=rectangle, minimum height=0.75cm, minimum width=0.75cm]
\tikzstyle{black dot}=[fill=black, draw=black, shape=circle, scale=0.5]
\tikzstyle{smallbox}=[fill=white, draw=black, shape=rectangle]
\tikzstyle{tri}=[fill=white, draw=black, regular polygon sides=3, regular polygon, shape border rotate=90, tikzit fill={rgb,255: red,0; green,128; blue,128}]
\tikzstyle{delay}=[fill=white, draw=black, tikzit fill=magenta, scale=0.7, and gate US]
\tikzstyle{downtri}=[fill=white, draw=black, tikzit fill={rgb,255: red,0; green,128; blue,128}, regular polygon, regular polygon sides=3, scale=0.7]
\tikzstyle{supertallbox}=[fill=white, draw=black, shape=rectangle, minimum height=2cm]
\tikzstyle{hint}=[fill=white, text=red, fill opacity=0, text opacity=1]
\tikzstyle{DDf}=[fill=white, draw=black, shape=rectangle, minimum height=3cm]

\tikzstyle{state}=[-, tikzit draw={rgb,255: red,0; green,187; blue,0}]
\tikzstyle{arrow}=[->, shorten <=0.1cm, shorten >=0.1cm]
\tikzstyle{focus}=[-, dashdotted, draw=red]
\tikzstyle{gray}=[-, draw={rgb,255: red,191; green,191; blue,191}]

\newcommand{\seq}[1]{\Bf{#1}}
\newcommand{\seqof}[1]{[#1]}
\newcommand{\chop}{\mathop{\bigcirc}}
\newcommand{\tr}[2]{\mathop{\mathrm{tr}}\nolimits^{#1}_{#2}}
\newcommand{\Tc}{\mathrm{Tc}}
\newcommand{\Un}{\mathrm{Un}}
\newcommand{\st}{\mathrm{st}}
\newcommand{\teq}{\mathbin{\triangleq}}

\newcommand{\prv}{\mathop{\mathrm{prv}}}
\newcommand{\nxt}{\mathop{\mathrm{nxt}}}
\renewcommand{\dom}{\mathop{\mathrm{dom}}}
\renewcommand{\cod}{\mathop{\mathrm{cod}}}
\newcommand{\sq}[1]{Dbl(#1)}
\newcommand{\caus}[1]{\mathrm{St}(#1)}
\newcommand{\causz}[1]{\mathrm{St}_0(#1)}

\newcommand{\cell}[5]{#1:#3\xrightarrow[#4]{#2}#5} 

\newcommand{\final}[1]{{!_{#1}}}
\newcommand{\Euci}{\Bf{Euc}_\infty}

\NewDocumentCommand{\valtostate}{m O{} O{}}{\,^{#2}\lfloor#1\rceil_{#3}}
\newcommand{\gc}[1]{{#1}_\bul} 
\newcommand{\seqel}[1]{\gc{\seq{#1}}}
\newcommand{\nextseqel}[1]{\seq{#1}_{\bul+1}}
\newcommand{\Nat}{\NN}
\newcommand{\tuple}[1]{\langle #1 \rangle}
\def\squD{\mc D}    
\def\seqD{\squD^*}  
\def\vc{;}  
\def\hc{*}  
\def\cc{\boxtimes}  

\title{Differentiable Causal Computations via\\Delayed Trace}

\author{
\IEEEauthorblockN{David Sprunger}
\IEEEauthorblockA{National Institute of Informatics
Tokyo, Japan 100-0003\\
Email: sprunger@nii.ac.jp}

\and 

\IEEEauthorblockN{Shin-ya Katsumata}
\IEEEauthorblockA{National Institute of Informatics
Tokyo, Japan 100-0003\\
Email: s-katsumata@nii.ac.jp}
}

\begin{document}

\maketitle

\begin{abstract}
  We investigate causal computations taking sequences of inputs to sequences
  of outputs where the $n$th output depends on the first $n$ inputs only. We
  model these in category theory via a construction taking a Cartesian
  category $\CC$ to another category $\caus\CC$ with a novel trace-like
  operation called ``delayed trace'', which misses yanking and dinaturality
  axioms of the usual trace. The delayed trace operation provides a feedback
  mechanism in $\caus\CC$ with an implicit guardedness guarantee.

  When $\CC$ is equipped with a Cartesian differential operator, we construct
  a differential operator for $\caus\CC$ using an abstract version of
  backpropagation through time, a technique from machine learning based on
  unrolling of functions. This obtains a swath of properties for
  backpropagation through time, including a chain rule and Schwartz theorem.
  Our differential operator is also able to compute the derivative of a stateful
  network without requiring the network to be unrolled.
\end{abstract}

\section{Introduction}\label{sec:intro}



Many objects of study in computer science, such as Mealy machines, clocked
digital circuits, signal flow graphs, discrete-time feedback loops, and
recurrent neural networks, compute a \emph{stateful} and particularly a
\emph{causal} function of their inputs, meaning the output of the function at
a particular time may depend on not only the current input, but also all
inputs received by the device up to that time. They share a basic operational
scheme, depicted in the following diagram (which is to be read left-to-right):
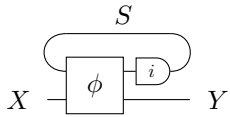
\ctikzfig{simplefeedback} 
Here the box labeled $\phi$ is a (sub)device which takes an $S$-value at its
upper left interface and an $X$-value at its lower left interface and produces
output $S$- and $Y$-values at its right interfaces. The differently-shaped box
labeled $i$ is our depiction of a \emph{delay gate}, a device which stores the
value provided to its left boundary and emits it one step later at its right
boundary, initially emitting the value $i$. The whole device, which we call
$\Phi$, receives sequences of $X$-valued inputs at the left and emits
sequences of $Y$-valued outputs at the right, storing its internal state in
the delay gate.

A recurrent neural network has inputs of two types: data inputs and
parameters. \emph{Training} a neural network means finding parameter values
$\theta$ so that when $\theta$ is fixed (in the diagram below by the
triangular device which emits $\theta$ constantly), the resulting function of
data inputs has a desired behavior.

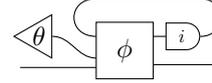
\ctikzfig{simplernn}

The key insight of \emph{gradient-based training} is that the derivative of
$\Phi$ with respect to $\theta$ gives an accurate prediction about how the
output of $\Phi$ will change in response to a small change in $\theta$,
allowing the trainer to make iterative small changes to $\theta$ to drive the
network to a desired behavior.

This idea works perfectly for feedforward (stateless) neural networks.
Recurrent neural networks require a workaround, however, due to the fact that
classical differentation does not work on stateful functions (or must be
performed in an infinite dimensional vector space).

The usual workaround is to first \emph{unroll} $\Phi$ into a
sequence of stateless functions, to which classical differentiation can be
applied \cite{Goodfellow16}. To be more precise, think of $\Phi$ as the
solution to the following recurrence relation:
$$(s_{k+1}, y_k) = \phi(s_k, x_k) \text{ where } s_0 = i.$$ 
Let $\phi_S = \pi_0 \circ \phi$ and $\phi_Y = \pi_1 \circ \phi$. Then the
unrolling of $\Phi$ is the sequence
\mbox{$\phi_k: X^{k+1} \to Y$} given by
\begin{align}
  \phi_0(x_0) &= \phi_Y(i, x_0) \nonumber \\
  \phi_1(x_0, x_1) &= \phi_Y(\phi_S(i, x_0), x_1) \label{eq:unr} \\
  \phi_2(x_0, x_1, x_2) &= \phi_Y(\phi_S(\phi_S(i, x_0), x_1), x_2) \nonumber \\
  \vdots \nonumber &
\end{align}
When the gradient of $\Phi$ is needed at an input of length $k$ by a trainer,
the gradient of $\phi_k$ at that input is used instead.

This is an empirically useful way to find gradients, known in the machine
learning literature as \emph{backpropagation through time} (BPTT) \cite{bptt}.
However, its ad-hoc nature raises some fundamental questions, the principal
one we address here being: \textbf{Does BPTT have the usual properties of
differentiation, or is it just a process involving differentiation?} That is,
does this unroll-then-differentiate procedure have a chain rule, a sum rule, a
notion of partial derivative, etc., or is it merely an empirically useful
process using derivatives?

We show that BPTT has the properties of differentiation mentioned above and
more. In particular, we are able to state the derivative of a stateful
function as another stateful function, rather than a sequence of stateless
functions. Roughly speaking, we accomplish this by taking advantage of the
fact that the unrolling above is an iterated composition of $\phi$ with
itself, and therefore its componentwise derivative can be ``re-rolled'' back
into a single stateful function.

\textbf{Outline.} Our first main contribution is to give a construction which
extends any given (Cartesian) category $\CC$, representing stateless
functions, to a new category $\caus\CC$ of stateful functions, particularly
computations extended through discrete time (definition~\ref{def:caus}). This
$\caus-$ construction captures causal functions as a special instance
(theorem~\ref{causal}), and captures other stateful devices like Mealy
machines and recurrent neural networks.

A distinctive feature of this construction includes the loop-with-delay gate
seen in the first diagram, which we will more formally call a {\em delayed
trace operator} (definition~\ref{def:dtr}). This delayed trace satisfies many
of the properties of its better-known cousin, the trace operator of Joyal et
al.~\cite{jsv} (proposition~\ref{prop:dtr_props}), but is missing the yanking
condition and satisfies a modified form of dinaturality
(theorem~\ref{thm:dinaturality}).

Our second major contribution is to give an abstract form of differentiation
in this category of stateful computations. A key result of this paper is that
if $\CC$ is a \emph{Cartesian differential category} \cite{cartesiandiffcat},
then so is $\caus\CC$ (theorem~\ref{pp:cartdiff}). In particular, this
differential operator matches the results obtained by
unrolling-then-differentiating as in BPTT (theorem~\ref{th:unrol}). The
definition of Cartesian differential categories packages many of the classic
properties of derivatives in a convenient abstract unit. Hence, showing that
$\caus\CC$ is a Cartesian differential category implicitly obtains a slew of
fundamental results for differentiation of stateful computations.

\section*{Related Work}


Signal flow graphs are a widely used model of causal computation, especially
in synchronous digital circuits and signal processing
\cite{1451723,Parhi2013}. The formation of loop paths in signal flow graphs
are often restricted so that each loop path must go through at least one
(initialized) delay gate. The delayed trace operator in $\caus\CC$ in this
paper embodies this principle.

A line of coalgebraic study of signal flow graphs by Rutten
\cite{DBLP:journals/mscs/Rutten05,DBLP:journals/lmcs/Rutten08}, Milius
\cite{DBLP:conf/lics/Milius10}, Hansen et
al. \cite{DBLP:journals/corr/HansenKR16} and Basold et
al. \cite{Basold2014} and many others achieve characterisations of
computable streams by signal flow graphs.  These coalgebraic
studies regard signal flow graphs as specification of coalgebraic
transition systems. This makes it possible to apply powerful
coalgebraic techniques to analyse the behaviour of signal flow
graphs. Our categorical work, on the other hand, regards signal flow
graphs as morphisms in a certain category, and focuses on the
categorical structures realising these flow
graphs.

An axiomatic system for representing digital circuits based on
monoidal category theory has been proposed by Ghica et
al. \cite{Ghica:2016:CSD:3077629.3077642,DBLP:conf/csl/GhicaJL17}.
Their system is an extension of a traced cartesian category with a few
structural morphisms that implement wire join and delay gate, but
their delay gates do not support arbitrary initialization. Their system can
represent interesting well-defined digital circuits using general loops
without delay gates. The precise relationship between their axiomatic
system and our categorical construction is not clear yet, and it is an
interesting topic to investigate.

\newcommand{\IH}{\mathbb{IH}} 

Zanasi studies the PROP $\IH_R$ of {\em interacting Hopf algebras} over a ring
$R$ in his PhD thesis \cite{DBLP:phd/hal/Zanasi15}. The expressive power of
this PROP is demonstrated by encoding various graphical systems into $\IH_R$
\cite{DBLP:conf/concur/BonchiSZ14,DBLP:conf/popl/BonchiSZ15}. When $R$ is the
polynomial ring $k[x]$ over a field $k$, the PROP $\IH_{k[x]}$ admits delay
gates, and the trace-with-delay operation (which he called {\em z-feedback
operator}) is definable \cite[Definition 7]{DBLP:conf/concur/BonchiSZ14}. His
z-feedback operator is very close to the delayed trace operator, except that
the latter supports arbitrary initial values. 

Recently, Kissinger and Uijlen reformulated the concept of causality in
quantam physics in a class of compact closed categories
\cite{DBLP:conf/lics/KissingerU17}. Starting from a compact closed
category with some extra structure, they refine it to the *-autonomous
category so that morphisms there respect causal constraints.



A category whose morphisms are realized by Mealy-machine like transducers is
constructed in the memoryful GoI by Hoshino et al.
\cite{DBLP:conf/csl/HoshinoMH14}. Their transducers, represented as functions
of type $S\times A\arrow T(S\times B)$, extend deterministic Mealy machines
with the ability to perform computational effects represented by the monad
$T$. The machine type considered in our work does not support these abstract
computational effects. Another technical difference from our work is that the
monoidal structure on their category of transducers is based on finite
coproducts in order to realize the particle-style trace operator for the GoI
interpretation, whereas our work uses finite products.

A common theme in recursively defined computations is that to have
well-defined behaviour, a recursive computation must satisfy a guardedness
condition \cite{DBLP:conf/icfp/AbelP13,DBLP:conf/csl/Mogelberg14}.  Goncharov
and Schr\"oder developed the theory of {\em guarded traced categories} to
formalize this phenomena in \cite{10.1007/978-3-319-89366-2_17}. The key idea
is to restrict Joyal et al.'s trace operator \cite{jsv} to a class of {\em
guarded morphisms}, which are an abstractly given class of morphisms
satisfying the guardedness condition. It is interesting to see the
relationship between guarded trace operator and the delayed trace operator,
and the key in this comparison is the treatment of the initial state, which is
missing in the guarded trace operator.

The idea of using tiles as representations of computation steps is pursued in
the {\em tile models} by Gadducci and Montanari
\cite{DBLP:conf/birthday/GadducciM00}. In their model, each tile $\cell f S A
{S'} B$ represents a state transition from $A$ to $B$, while $S$ and $S'$ are
the trigger of and effect of this transition, respectively. In our work, $S$
and $S'$ denote types of values stored across clock ticks.


Inspired by the semantics of differential $\lambda$-calculus and
differential proof nets by Ehrard and Regnier
\cite{DBLP:journals/tcs/EhrhardR03,DBLP:journals/entcs/EhrhardR05},
Blute, Cockett and Seely categorically formalized the differentiation
operator in analysis. The formalization was first given in the
categories where morphisms denote linear maps
\cite{DBLP:journals/mscs/BluteCS06}. Later, they introduced a new
axiomatization \cite{cartesiandiffcat} based on cartesian monoidal
category where morphisms denote possibly non-linear maps.  This paper
is based on the latter work, and adopts more recent reformulations of
differentiation operators studied in \cite{cruttwell_2017} and
\cite{sdg2014}. 

There have been some recent efforts to connect category theory with machine
learning, particularly backpropagation, using the fact that differentiation
has a chain rule and is therefore compositional, for example \cite{FongST17}.
A notable example is \cite{DBLP:journals/pacmpl/Elliott18}, where Elliot
studies automatic differentiation (AD) in the context of functional
programming. He gives a clean account of an AD algorithm by exploiting the
functorial nature of the differentiation operator, including both a chain
rule and a parallel rule to obtain a Cartesian functor.

\section*{Preliminaries}

We assume familiarity with basic category theory. If $\CC$ is a category, we
write $|\CC|$ to denote its objects, and $\CC(X, Y)$ to denote a homset for
$X, Y \in |\CC|$. We may abbreviate an identity map $\id_X$ to the name of its
object, $X$.

If $\CC$ is a cartesian category, we write $1$ for its terminal object,
$!_X:X\arrow 1$ for the unique maps to $1$, and $\times$ for the product
bifunctor. The tupling of morphisms $f_i:Y\arrow X_i$ for $i\in\{0,1\}$ is denoted by
$\tuple{f_0,f_1}$. Projections are denoted by $\pi^{X_0,X_1}_i:X_0\times
X_1\arrow X_i$ ($i\in\{0,1\}$), and we drop the superscript when it is obvious from
context. The symmetry map on products is $\sigma_{X, Y}: X\times Y \to
Y\times X$.

In general, Cartesian categories need not be strict, but working with
associators etc.~unnecessarily complicates the story. So whenever we mention a
Cartesian category, we will instead technically be using the equivalent
strictified version.

Bold metavariables---$\seq X,\seq s$, etc.---denote sequences of mathematical
objects, indexed by $\Nat$. The $i$th component of a sequence is $\seq X_i$.
By $\chop{\seq X}$ we mean the tail of $\seq X$, namely $\chop\seq X=\seq
X_1,\seq X_2,\cdots$. In addition to Roman-letter subscripts, we use a bullet
$\bul$ as an special index variable, which can be bound by the
sequence-forming bracket notation given next.

Let $e$ be an expression containing some
dotted sequence metavariables $\seq X_\bul,\seq Y_\bul\cdots$.  By
$\seqof{e}$ we mean the infinite sequence obtained by substituting
$0,1,2\cdots$ for $\bul$. For instance,
\begin{align*}
  (i,\seqof{\seq x_\bul+\seq y_\bul}) & ~\text{is}~ (i, \seq x_0+\seq y_0,~\seq x_1+\seq y_1,~\cdots) \\
  \seqof{(i,\seq x_\bul+\seq y_\bul)} & ~\text{is}~ (i,\seq x_0+\seq y_0),~(i,\seq x_1+\seq y_1),~\cdots
\end{align*}
When $e$ contains at least one dotted sequence metavariable, we may omit the
outermost $\seqof-$, so $\seqof{\seq X_\bul\times \seq Y_\bul}$ may be written
as $\seq X_\bul\times \seq Y_\bul$.  This omission is not allowed when $e$
contains no such variable; otherwise we would confuse ordinary expressions
(like $x+y$) and constant infinite sequences (like $[x+y]=x+y,x+y,\cdots$).

A mathematical formula $\phi$ containing dotted sequence metavariables
represents the conjunction $\bigwedge_{i\in\Nat}\phi[i/\bul]$.  For
instance, $\gc{\seq Z}=\gc{\seq X}\times\gc{\seq Y}$ means
$\fa{i\in\Nat}\seq Z_i=\seq X_i\times\seq Y_i$.

\section{Extending Cartesian Categories along Discrete Time}

Before jumping into the depths of categorical abstraction, we take a moment to
think about different kinds of functions on sequences and particularly where
causal functions lie.

One natural way to obtain functions on sequences is to consider the category
$\Set^\NN$, the countable product category of $\Set$. In this category, each
morphism $\seq f:\seq X\arrow\seq Y$ consists of independent components $\seq
f_k:\seq X_k\arrow\seq Y_k$ for all $k \in \NN$, each of which compute a
single entry in the output sequence.

These are certainly functions taking sequences to sequences in a causal
manner, but the fact that each of the components of $\seq f$ are independent
means the $k$th output of $\seq f$ depends only on the $k$th input, not on all
inputs before $k$. Therefore, some causal functions of sequences, such as
computing a running average, are missing from this class.

Another natural idea would be to take all the functions in homsets
$\Set(\prod\seq X, \prod\seq Y)$ for arbitrary $\seq X, \seq Y \in |\Set^\NN|$. 
This class is too big---non-causal functions such as 
${\tt tl}: (x_0, x_1, \ldots) \mapsto (x_1, x_2, \ldots)$ are present
there. Therefore, we must do something a bit more complex to obtain
a class of functions somewhere between these two.



To obtain the class of causal functions, we return to our original idea,
$\Set^\NN$, and add objects in the domain and codomain of each component of
$\seq f$ representing communication channels with its neighbouring components,
like
\begin{equation}
  \label{eq:fk}
  \seq f_k:\seq S_k\times\seq X_k\arrow\seq S_{k+1}\times\seq Y_k.  
\end{equation}
To start this computation, we need to provide an initial state $i:1\arrow\seq
S_0$, and we call the pair $(i,\seq f)$ a {\em stateful morphism sequence}. We
will see causal functions are equivalence classes of these stateful morphism
sequences (theorem~\ref{causal}).

Though these are all functions on sequences, it will often be convenient to
pretend that these sequences are produced one element at a time, synchronized
by some clock signal. Thus, since the function $\seq f_k$ above computes the
$k$th element in the sequence, we may refer to it as producing a value
\emph{at clock tick $k$}. Similarly, we refer to the element of state passed
from $\seq f_k$ to $\seq f_{k+1}$ as being kept \emph{across clock ticks}, and
other such language. In this way, computing functions of sequences can also
be thought of as performing discrete timed computations.


There is a clear distinction between the role of $\seq X_k/\seq Y_k$ and $\seq
S_k/\seq S_{k+1}$---the former objects are the types of {\em values} flowing
through $\seq f_k$ at clock tick $k$, while the latter objects are the types of
{\em states} passed across clock ticks. We organize these two different kinds of
information flow using special two-dimensional categories called {\em double
categories} \cite{Ehresmann1963}. 

Roughly speaking, double categories consist of 0-cells (objects), two types of
1-cells (horizontal and vertical morphisms), and 2-cells (tiles) which go
between pairs of horizontal and vertical 1-cells. These 2-cells are often
drawn like below (left).
\begin{displaymath}
  \xymatrix{
    X \rrh S \rdm A \rmid{rd}{\alpha} & Y \rdh B & &
    \cdot \rrh {\seq S_k} \rdm {\seq X_k} \ar@{}[rd]|-{\seq f_k} & \cdot \rdh {\seq Y_k} \\
    Z \rrm {S'} & W & &
    \cdot \rrm {\seq S_{k+1}} & \cdot 
  }
\end{displaymath}
These tiles can be composed along either common vertical 1-cells
(\emph{horizontal composition}) or common horizontal 1-cells (\emph{vertical
composition}). Having these two distinct types of composition is the essential
and only reason for using a double category in this paper, so that we can use
one composition for composition \emph{within} a clock tick and the other for
composition \emph{across} clock ticks. We will not be using any results of
higher category theory or further higher-dimensional abstractions.

Our double category will therefore have a particularly simple structure, with
2-cells as above (right). We have a dummy 0-cell ($\cdot$), objects from $\CC$
as 1-cells, representing values when oriented vertically and states when
oriented horizontally, and functions $\seq f_k$ on states and values in the
tiles.

\begin{definition}\label{def:double_cat}
  Let $(\CC, \times, 1)$ be a (strict) Cartesian category.  The double
  category $\sq{\CC}$ is defined as follows:
  \begin{itemize}
  \item $\cdot$ is the only object (0-cell)
  \item Horizontal and vertical 1-cells are both given by objects of
    $\CC$, composed with $\times$, and have $1$ as the identity.
  \item A 2-cell $f$ with source horizontal 1-cell $S$, source
    vertical 1-cell $X$, target horizontal 1-cell $S'$, and target
    vertical 1-cell $Y$ is a morphism
    \mbox{$\phi \in \CC(S\times X, S'\times Y)$}.
    \begin{displaymath}
      \xymatrix{
        \rrh{S} \rdm{X} \cdot \ar@{}[rd]|-f & \rdh{Y\ \ =} \cdot \\
        \rrm{S'} \cdot & \cdot
      }
      \xymatrix{
        \rrh{\prv f} \rdm{\dom f} \cdot \ar@{}[rd]|-f & \rdh{\cod f} \cdot \\
        \rrm{\nxt f} \cdot & \cdot
      }
    \end{displaymath}
    As indicated above, we denote the source and target 1-cells of
    $f$---$S, X, S',$ and $Y$---by $\prv f, \dom f, \nxt f,$ and
    $\cod f$, respectively.  We will generally denote a 2-cell by
    $\cell f S X {S'} Y$. We call $\phi$ the \emph{underlying
      morphism} of $f$, while $U$ is the operation taking a 2-cell to
    its underlying morphism, so $Uf = \phi$.
  \item The horizontal composition of 2-cells, say the $f$ above
    before $\cell g T Y {T'} Z$, is
    $\cell{g\hc f}{S\times T} X {S'\times T'}Z$ with underlying
    morphism
    \begin{displaymath}
      U(g \hc f) \teq (S' \times Ug)\circ(\sigma_{S', T}\times Y)\circ(T \times Uf)\circ(\sigma_{S, T}\times X).
    \end{displaymath}
  \item The vertical composition of 2-cells, say the $f$ above before
    $\cell h {S'} V {S''} W$, is
    $\cell{f \vc h} S {X\times V} {S''} {Y\times W}$ with underlying
    morphism
    \begin{displaymath}
      U(f\vc h) \teq (S'' \times\sigma_{Y, W})\circ(Uh\times Y)\circ(S'\times\sigma_{V,Y})\circ(Uf\times V).
    \end{displaymath}
  \end{itemize}
  NB: vertical composition is given in relational composition order
  while horizontal composition is given in functional composition
  order.
\end{definition}

String diagrams for the underlying $\CC$-morphisms of horizontal and vertical
composites may be helpful to digest this definition. The underlying morphism
of the horizontal composition $g\hc f$ is:

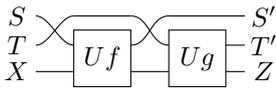
\ctikzfig{horcomp}

While for vertical composition, we have $U(f\vc h)$ as below:

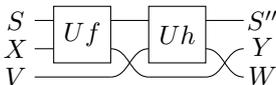
\ctikzfig{vertcomp}

A 2-cell $f$ of $\sq\CC$ is determined by its underlying morphism $\phi$ from
$\CC$. To stress this, we often draw $\phi$ inside the tile,
with its inputs and outputs connected to corresponding edges: 
\ctikzfig{2cell} 

Horizontal composition of 2-cells is composition along values like $\seqel X$
or $\seqel Y$, and we think of as occuring within a single clock tick.
Vertical composition is composition along states like $\seqel S$, and occurs
across clock ticks.

\begin{definition}
  For $\phi\in\CC(X,Y)$, the 2-cells \mbox{$\cell{\phi^h} 1 X 1 Y$}
  and \mbox{$\cell{\phi^v} X 1 Y 1$} have
  $U(\phi^h) \teq \phi \teq U(\phi^v) $.
\end{definition}

These operations sending $\CC$-morphisms to 2-cells in $\sq\CC$ are
particularly useful. (Note first that $\id_X^h$ are the identities for
horizontal composition, and similarly $\id_S^v$ are the identities for
vertical composition!) More practically, 2-cells of the form $\phi^h$
modify values only, while 2-cells of the form $\phi^v$ modify states
only, as shown in the following lemma.

\begin{lemma}\label{lem:adjust}
  If $\cell f S X {S'} Y$ is a $\sq\CC$ 2-cell, 
  \mbox{$\phi_1 \in \CC(T, S)$}, \mbox{$\phi_2 \in \CC(S', T')$},
  \mbox{$\psi_1 \in \CC(W, X)$}, and \mbox{$\psi_2 \in \CC(Y, Z)$},
  then the underlying morphism of
  \mbox{$\phi_1^v\vc(\psi_2^h\hc f \hc \psi_1^h)\vc \phi_2^v$} has the
  following string diagram in $\CC$: \ctikzfig{allshims}
\end{lemma}

Note that if $\phi_2$ and $\psi_2$ are identities, the composed 2-cell
above is denoted $\phi_1^v\vc f\hc \psi_1^h$. This compact notation
for precomposition in both dimensions is a powerful notational
advantage of having the $\vc$ and $\hc$ operators take their arguments
in different orders.

\subsection{Stateful Morphism Sequences and Extensional Equivalence}

Each 2-cell of the double category $\sq\CC$ represents an individual component
computing a single output value in a time-extended computation, like that of
\cref{eq:fk}. To represent a whole causal computation, we collect together
countably many of these components into a \emph{stateful morphism sequence}.

\begin{definition}
  \label{def:ssm}
  Let $\seq X,\seq Y\in|\CC|^\Nat$ be sequences of $\CC$-objects.
  A {\em stateful morphism sequence} of type $\seq X\arrow\seq Y$ is a
  pair $(i,\seq s)$ of a sequence $\seq s$ of 2-cells in $\sq\CC$ and
  a $\CC$-morphism $i :1\arrow\prv\seq s_0$ such that
  \begin{displaymath}
    \dom{\seq s_\bul}=\seq X_\bul,\quad
    \cod{\seq s_\bul}=\seq Y_\bul,\quad
    \nxt{\seq s_\bul}=\prv{\seq s_{\bul+1}}.    
  \end{displaymath}
  The \emph{state sequence} of $(i,\seq s)$ is
  \mbox{$\st(i,\seq s) \teq \seqof{\prv(\seq s_\bul)}$}.
\end{definition}

Note the last condition implies $\seq s_\bul\vc \seq s_{\bul+1}$ exists, which
allows each component to pass state to the next.

A stateful morphism sequence can be thought of as an infinite tower of
2-cells, each layer of which is vertically composable with adjacent layers, as
depicted in \cref{fig:sms} left.

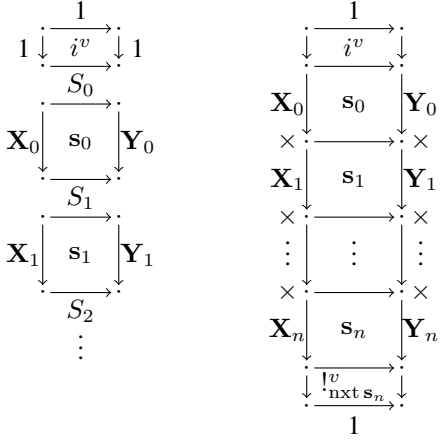
\begin{figure}[h!]
  \begin{center}
    \begin{tikzpicture}
	\begin{pgfonlayer}{nodelayer}
		\node [style=none] (0) at (0, 2) {$\cdot$};
		\node [style=none] (1) at (2, 2) {$\cdot$};
		\node [style=none] (2) at (0, 0) {$\cdot$};
		\node [style=none] (3) at (2, 0) {$\cdot$};
		\node [style=none] (4) at (1, 1) {$\seq s_0$};
		\node [style=none] (5) at (-0.5, 1) { $\seq X_0$};
		\node [style=none] (6) at (2.5, 1) { $\seq Y_0$};
		\node [style=none] (7) at (1, 2.5) { $S_0$};
		\node [style=none] (8) at (1, -0.5) { $S_1$};
		\node [style=none] (9) at (0, -1) {$\cdot$};
		\node [style=none] (10) at (2, -1) {$\cdot$};
		\node [style=none] (11) at (0, -3) {$\cdot$};
		\node [style=none] (12) at (2, -3) {$\cdot$};
		\node [style=none] (13) at (1, -2) {$\seq s_1$};
		\node [style=none] (14) at (-0.5, -2) { $\seq X_1$};
		\node [style=none] (15) at (2.5, -2) { $\seq Y_1$};
		\node [style=none] (16) at (1, -3.5) { $S_2$};
		\node [style=none] (25) at (0, 4) {$\cdot$};
		\node [style=none] (26) at (2, 4) {$\cdot$};
		\node [style=none] (27) at (0, 3) {$\cdot$};
		\node [style=none] (28) at (2, 3) {$\cdot$};
		\node [style=none] (29) at (1, 3.5) {$i^v$};
		\node [style=none] (30) at (-0.5, 3.5) { 1};
		\node [style=none] (31) at (2.5, 3.5) { 1};
		\node [style=none] (32) at (1, 4.5) { 1};
		\node [style=none] (33) at (1, -4.25) { \vdots};
	\end{pgfonlayer}
	\begin{pgfonlayer}{edgelayer}
		\draw [style=arrow] (0.center) to (1.center);
		\draw [style=arrow] (1.center) to (3.center);
		\draw [style=arrow] (0.center) to (2.center);
		\draw [style=arrow] (2.center) to (3.center);
		\draw [style=arrow] (9.center) to (10.center);
		\draw [style=arrow] (10.center) to (12.center);
		\draw [style=arrow] (9.center) to (11.center);
		\draw [style=arrow] (11.center) to (12.center);
		\draw [style=arrow] (25.center) to (26.center);
		\draw [style=arrow] (26.center) to (28.center);
		\draw [style=arrow] (25.center) to (27.center);
		\draw [style=arrow] (27.center) to (28.center);
	\end{pgfonlayer}
\end{tikzpicture}
    \hspace{1cm}
    \begin{tikzpicture}
	\begin{pgfonlayer}{nodelayer}
		\node [style=none] (0) at (1, 3) {};
		\node [style=none] (1) at (3.5, 3) {};
		\node [style=none] (2) at (1, 1) {$\cdot$};
		\node [style=none] (3) at (3.5, 1) {$\cdot$};
		\node [style=none] (4) at (2.25, 2) {$\seq s_0$};
		\node [style=none] (5) at (0.5, 2) { $\seq X_0$};
		\node [style=none] (6) at (4, 2) { $\seq Y_0$};
		\node [style=none] (9) at (1, 1) {};
		\node [style=none] (10) at (3.5, 1) {};
		\node [style=none] (11) at (1, -1) {$\cdot$};
		\node [style=none] (12) at (3.5, -1) {$\cdot$};
		\node [style=none] (13) at (2.25, 0) {$\seq s_1$};
		\node [style=none] (14) at (0.5, 0) { $\seq X_1$};
		\node [style=none] (15) at (4, 0) { $\seq Y_1$};
		\node [style=none] (17) at (1, 4) {$\cdot$};
		\node [style=none] (18) at (3.5, 4) {$\cdot$};
		\node [style=none] (19) at (1, 3) {$\cdot$};
		\node [style=none] (20) at (3.5, 3) {$\cdot$};
		\node [style=none] (21) at (2.25, 3.5) {$i^v$};
		\node [style=none] (22) at (0.5, 1) { $\times$};
		\node [style=none] (24) at (2.25, 4.5) { 1};
		\node [style=none] (26) at (1, -5) {};
		\node [style=none] (27) at (3.5, -5) {};
		\node [style=none] (28) at (1, -6) {$\cdot$};
		\node [style=none] (29) at (3.5, -6) {$\cdot$};
		\node [style=none] (30) at (2.25, -5.5) {$!_{\nxt \seq s_n}^v$};
		\node [style=none] (33) at (1, -1) {};
		\node [style=none] (34) at (3.5, -1) {};
		\node [style=none] (35) at (1, -3) {$\cdot$};
		\node [style=none] (36) at (3.5, -3) {$\cdot$};
		\node [style=none] (37) at (2.25, -1.75) {\vdots};
		\node [style=none] (38) at (1, -3) {};
		\node [style=none] (39) at (3.5, -3) {};
		\node [style=none] (40) at (1, -5) {$\cdot$};
		\node [style=none] (41) at (3.5, -5) {$\cdot$};
		\node [style=none] (42) at (2.25, -4) {$\seq s_n$};
		\node [style=none] (43) at (0.5, -4) { $\seq X_n$};
		\node [style=none] (44) at (4, -4) { $\seq Y_n$};
		\node [style=none] (45) at (2.25, -6.5) { 1};
		\node [style=none] (46) at (4, -1.75) {\vdots};
		\node [style=none] (47) at (0.5, -1.75) {\vdots};
		\node [style=none] (48) at (4, 1) { $\times$};
		\node [style=none] (49) at (0.5, -1) { $\times$};
		\node [style=none] (50) at (4, -1) { $\times$};
		\node [style=none] (51) at (0.5, -3) { $\times$};
		\node [style=none] (52) at (4, -3) { $\times$};
	\end{pgfonlayer}
	\begin{pgfonlayer}{edgelayer}
		\draw [style=arrow] (1.center) to (3.center);
		\draw [style=arrow] (0.center) to (2.center);
		\draw [style=arrow] (2.center) to (3.center);
		\draw [style=arrow] (10.center) to (12.center);
		\draw [style=arrow] (9.center) to (11.center);
		\draw [style=arrow] (11.center) to (12.center);
		\draw [style=arrow] (17.center) to (18.center);
		\draw [style=arrow] (18.center) to (20.center);
		\draw [style=arrow] (17.center) to (19.center);
		\draw [style=arrow] (19.center) to (20.center);
		\draw [style=arrow] (27.center) to (29.center);
		\draw [style=arrow] (26.center) to (28.center);
		\draw [style=arrow] (28.center) to (29.center);
		\draw [style=arrow] (34.center) to (36.center);
		\draw [style=arrow] (33.center) to (35.center);
		\draw [style=arrow] (35.center) to (36.center);
		\draw [style=arrow] (39.center) to (41.center);
		\draw [style=arrow] (38.center) to (40.center);
		\draw [style=arrow] (40.center) to (41.center);
	\end{pgfonlayer}
\end{tikzpicture}
  \end{center}
  \caption{The stateful morphism sequence $(i, \seq s)$ and its $n$th truncation}
  \label{fig:sms}
\end{figure}

In this representation, the ``arrow of time'' starts at $\seq X_0$ and
points down. At the zeroth clock tick, the stateful morphism sequence
receives a value at $\seq X_0$, outputs a value at $\seq Y_0$, and
sets a state value of type $S_1$. Then at the first clock tick, the
first layer of the stateful morphism sequence executes, using the state
previously prepared by the zeroth layer.

Since we intend the state maintained by these sequences to be
internal, saying $(i,\seq s) = (j,\seq t)$ if and only if they are
exactly the same sequence of 2-cells is not a suitable notion
of equality. Ideally, if we could form the infinite vertical
composition of 2-cells, the natural definition of equality of two
stateful morphism sequences of type $\seq X\arrow\seq Y$ would be to
compare the underlying $\CC$-morphisms of the infinite composition,
meaning $(i, \seq s) = (j, \seq t)$ if and only if
\begin{displaymath}
  U(i^v;\seq s_0;\seq s_1;\cdots)=U(j^v;\seq t_0;\seq t_1;\cdots):
  \prod_{i\in\Nat}\seq X_i\arrow \prod_{i\in\Nat}\seq Y_i.
\end{displaymath}
However, formalizing this infinite vertical composition is technically
challenging, and $\CC$ may not admit such countable products. We therefore
instead require that all finite initial segments of the sequence match using a
{\em truncation operation}.
\begin{definition}
  The \emph{$n$th truncation} of a stateful morphism sequence
  $(i,\seq s):\seq X\arrow\seq Y$ is the $\CC$-morphism
  \begin{align*}
    \Tc_n(i,\seq s)
    &\teq
      U(i^v \vc \seq s_0 \vc \cdots \vc \seq s_n\vc !_{\nxt \seq s_n}^v)
      :
      \prod_{k=0}^n\seq X_k\arrow \prod_{k=0}^n\seq Y_k.
  \end{align*}
\end{definition}

Graphically, $\Tc_n(i,\seq s)$ is the underlying morphism of the vertical
composite 2-cell depicted in \cref{fig:sms} right.\footnote{In \cref{fig:sms},
the 2-cells on the right have been drawn with common horizontal 1-cells
and vertical 1-cells composed with $\times$ to indicate the 2-cells have been
composed vertically, whereas on the left the 2-cells are separate since
the full infinite vertical composition may not be possible.}

\begin{definition}
  Two stateful morphism sequences
  $(i,\seq s),(j,\seq t):\seq X\arrow\seq Y$ are \emph{extensionally
    equal} iff $\Tc_\bul(i, \seq s) = \Tc_\bul(j, \seq t)$.
\end{definition}
It is easy to verify that extensional equality between stateful morphism
sequences is an equivalence relation.

The state sequences of extensionally equivalent stateful morphisms sequences
can be different, which is good because it matches our intention and bad
because it can be harder to decide whether two computation sequences are
equal. Comparing truncations is always possible, but sometimes technically
difficult. The following lemma has proven a useful method for establishing
extensional equality.

\begin{lemma}[Shim lemma]\label{lem:shim}\footnote{A shim is a little
    piece of material used to align two items, such as a sliver of wood
    between a door frame and surrounding wall studs. In this case, $\seq b$ is
    the shim, and it adjusts the state spaces of the two stateful morphism
    sequences.}  Suppose $(i, \seq s),(j, \seq t):\seq X\arrow\seq Y$ are
    stateful morphism sequences, and $\seq b$ is a sequence of $\CC$-morphisms
    such that $\seq b_\bul:\prv(\seq s_\bul)\arrow \prv(\seq t_\bul)$,
  \begin{displaymath}
    \seq b_0\circ i = j,\text{ and } \seq s_\bul\vc \seq b_{\bul+1}^v = \seq b_\bul^v \vc \seq t_\bul.
  \end{displaymath}
  Then $(i, \seq s)$ and $(j, \seq t)$ are extensionally equivalent.
\end{lemma}

\begin{proof}
  Show by induction that
  $$i^v\vc\seq s_0\vc\cdots\vc\seq s_n\vc \seq b_{n+1}^v\vc !_{\nxt{\seq
        t_n}}^v = j^v\vc\seq t_0\vc\cdots\vc \seq t_n\vc !_{\nxt{\seq
        t_n}}^v$$
\end{proof}

Unrolling and truncation are related operations,
and in fact we can extend unrolling to general stateful morphism sequences.
\begin{definition}
  Let $(i,\seq f):\seq A\arrow\seq B$ be a stateful morphism
  sequence. Its {\em $k$-th unrolling} is the $k$th projection of the
  $k$th truncation:
  $ \Un_k(i,\seq f)\teq \pi_k\circ \Tc_k(i,\seq f):\prod_{n=0}^k\seq
  A_n\arrow \seq B_k $.
\end{definition}
For instance, the recurrently defined functions $\phi_k$ in
\cref{sec:intro} are unrollings of a certain stateful morphism sequence
involving $\phi$ and $i$.

\subsection{Category of Causal Morphisms}

We are ready to construct our category of causal morphisms using stateful
morphism sequences and extensional equality between them.

\begin{definition}\label{def:idcompseq}
  The {\em identity stateful morphism sequence} $\id_{\seq X}$ is
  $(\id_1,\seqof{(\id_{\seq X_\bul})^h})$ for all
  $\seq X \in |\CC|^\Nat$.

  The \emph{composition} $(i,\seq s)\circ(j,\seq t)$ of stateful morphism sequences
 $(i,\seq s):\seq Y\arrow\seq Z$ and $(j,\seq t):\seq X\arrow\seq Y$
  is
  \begin{displaymath}
    (i,\seq s)\circ(j,\seq t)\teq (\langle j, i\rangle, \seqof{\gc{\seq s} \hc \gc{\seq t}}):\seq X\arrow\seq Z.
  \end{displaymath}
\end{definition}%
As usual, we may denote $\id_\seq X$ by $\seq X$.

In our ``tower of 2-cells'' representation, the composition of
stateful morphism sequences is in \cref{fig:comp}.
\begin{figure}[h!]
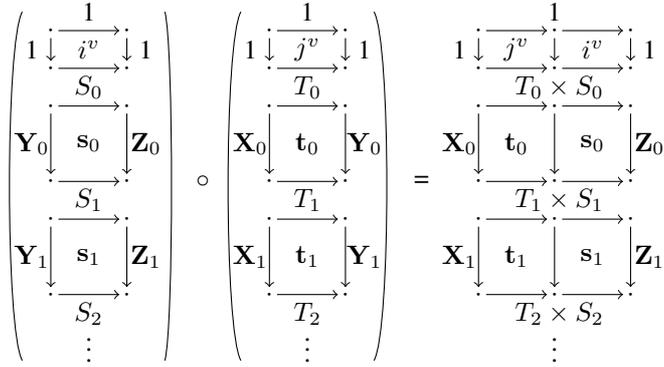

  \ctikzfig{sms_composition}
  \caption{Composition of stateful morphism sequences}
  \label{fig:comp}
\end{figure}
Note the state sequence of the composite is the componentwise product
of the original state sequences.

\begin{lemma}\label{lem:sms_cat}
  Composition of stateful morphism sequences is well-defined on
  extensional equivalence classes. Further, (the extensional
  equivalence class of) $\id_{\seq X}$ is the unit for the composition
  operation.
\end{lemma}

\begin{definition}\label{def:caus}
  Given a strict Cartesian category $\CC$, its \emph{causal extension}
  is a category $\caus{\CC}$ where
  \begin{itemize}
  \item objects are $|\CC|^\Nat$, that is, $\NN$-indexed families of
    $\CC$-objects,
  \item morphisms are extensional equivalence classes of stateful
    morphism sequences,
  \item identities and composition are the extensions of those in
    Definition~\ref{def:idcompseq} to the extensional equivalence
    classes by Lemma~\ref{lem:sms_cat}.
  \end{itemize}
\end{definition}
We will justify our use of the word ``causal'' by establishing a 
connection to the
existing notion of causal functions in \cref{causal}, but first we 
establish some properties of $\caus\CC$.

The category $\CC^\Nat$ is naturally included into $\caus\CC$ via the
functor $H:\CC^\Nat\arrow\caus\CC$:
\begin{align*}
  H\seq X&=\seq X,
  &
    H\seq f&=(\id_1,\seqof{\seq f_\bul^h}).
\end{align*}
We call the morphisms in $\caus\CC$ of the form $H\seq f$ {\em stateless
morphisms}, since they can be realized by a stateful morphism sequence with
state sequence $\seqof{1}$.\footnote{This looks like a citation, but it means
the constant sequence consisting of the terminal object of $\CC$ in every
position.}

\begin{proposition}
  $\caus{\CC}$ is Cartesian, and $H$ is finite-product preserving.
\end{proposition}
\begin{proof}
  In $\caus\CC$, the final object is $\seqof 1$ and the final map from
  $\seq X$ is $H\seqof{\final{\seq X_\bul}}$.  Products and projection
  are also componentwise: our chosen $\caus\CC$ product
  $\seq X\times\seq Y$ is the sequence
  $[\gc{\seq X}\times\gc{\seq Y}]$ of $\CC$ products, with
  $\pi_n \teq H\seqof{\pi_n^{\seq X_\bul,\seq Y_\bul}}$ for
  $n \in \{0, 1\}$.
\end{proof}

\subsection{Morphisms in $\caus\Set$ and Causal Functions}

We claim that morphisms of $\caus\CC$ represent {\em causal} computations,
whose outputs depend only on past inputs and states. To justify this claim, we
compare $\Set$-theoretic causal functions and morphisms in $\caus{\Set}$. For
this, we need a precise definition of causality for functions on sequences,
which we adapt from \cite{DBLP:journals/corr/HansenKR16}. First, for $\seq x,
\seq y\in A^\Nat$, by $\seq x\equiv_n\seq y$ we mean $\seq x$ and $\seq y$
match in the first $n$ positions, that is, $\seq x_i=\seq y_i$ holds for any
$i\le n$.
\begin{definition}[\cite{DBLP:journals/corr/HansenKR16}]
  Let $A$ and $B$ be sets. A function
  \mbox{$f:A^\Nat\arrow B^\Nat$} is {\em causal} if for any $\seq x,\seq y\in A^\Nat$,
  \begin{displaymath}
    \fa{n\in\Nat}
    \seq x\equiv_n\seq y\implies
    f(\seq x)\equiv_nf(\seq y).
  \end{displaymath}
\end{definition}

The following theorem states that $\caus\Set$ characterises causal
functions on streams.
\begin{theorem}\label{causal}
  The homset $\caus\Set(\seqof A,\seqof B)$ bijectively corresponds to
  the set of causal functions from $A^\Nat$ to $B^\Nat$.
\end{theorem}

The proof can be found in the appendix.

\subsection{The Category $\causz\CC$ and Deterministic Mealy Machines}

The input, output, and state types for a $\caus\CC$ morphisms can vary over
time. This is a crucial property to capture all causal functions, as seen in
the proof of \cref{causal}. However, the computational models we mentioned in
the introduction, like Mealy machines, are more regular, having fixed input,
output, and state types, and additionally executing the same function at each
time step. Thus it may appear we have overgeneralized. Luckily, we can recover
these regular causal functions in a subcategory of $\caus\CC$:
\begin{definition}
  The subcategory $\causz\CC$ of $\caus\CC$ has:
  \begin{itemize}
    \item objects of the form $\seqof X$ for some $X\in\CC$, and
    \item morphisms the (extensional equivalence classes of) stateful morphism sequences of the form
      $(i,\seqof f)$ for some 2-cell $\cell f S X S Y$.
  \end{itemize}
\end{definition}
It is easy to check that this restricted class of morphisms is closed under
the $\caus\CC$-composition, hence $\causz\CC$ is a well-defined subcategory.
We note the Cartesian structure of $\caus\CC$ restricts to $\causz\CC$.

\begin{proposition}
  The category $\causz\CC$ is Cartesian, and the functor
  $H_0:\CC\arrow\causz\CC$ is finite-product preserving.
  \begin{align*}
    H_0X&=\seqof X & H_0f&=(\id_1,\seqof{f^h}).
  \end{align*}
\end{proposition}

Morphisms of $\causz\Set$ may be identified as the causal functions that can
be computed by {\em deterministic Mealy machines}. Suppose $(i,\seqof
f):\seqof X\arrow\seqof Y$ is a morphism in $\causz\Set$. The set $S=\cod i$
is the set of states of the Mealy machine, $i:1\arrow S$ is the initial state,
and the function $f:S\times X\arrow S\times Y$ is the deterministic
transition-and-output function computing the next state and output from the
current state and input. The composition of morphisms in $\causz\CC$
corresponds to the {\em series (cascade) composition} of Mealy machines.

One useful operation on stateful morphism sequences is {\em unrolling}.
\begin{definition}
  Let $(i,\seq f):\seq A\arrow\seq B$ be a stateful morphism
  sequence. Its {\em $k$-th unrolling} is the $k$th projection of the
  $k$th truncation:
  $ \Un_k(i,\seq f)\teq \pi_k\circ \Tc_k(i,\seq f):\prod_{n=0}^k\seq
  A_n\arrow \seq B_k $.
\end{definition}
For instance, the recurrently defined functions $\phi_k$ in
\cref{eq:unr} in \cref{sec:intro} are unrollings:
$\phi_k=\Un_k(i,\seqof\phi)$.

Note that the truncation operation $\Tc$ can be extended to
$\caus\CC$-morphisms, as it is well-defined on extensional equivalence classes.

\section{Delayed Trace Operator}

The category $\caus{\CC}$ carries interesting structure that may not be
present in $\CC$---it has a {\em delayed trace operator}.  This is related to
Joyal et al.'s {\em trace operator} \cite{jsv}, which we briefly recall here.
The trace operator is a structure on braided monoidal categories, and is a
collection of functions $tr^S:\CC(S\ox X,S\ox Y)\arrow\CC(X,Y)$. In the
language of string diagrams, this operation is understood to form a feedback
loop at a specified pair of ports:
\begin{displaymath}
  \infer{tr^S(f):X\arrow Y \quad \begin{tikzpicture}
	\begin{pgfonlayer}{nodelayer}
		\node [style=box] (0) at (0, 0.125) {$f$};
		\node [style=none] (1) at (0.75, 0.5) {};
		\node [style=none] (2) at (0.75, -0.25) {};
		\node [style=none] (3) at (-0.75, 0.5) {};
		\node [style=none] (4) at (-0.75, -0.25) {};
		\node [style=none] (5) at (-1.5, -0.25) {};
		\node [style=none] (6) at (1.5, -0.25) {};
		\node [style=none] (8) at (-1.5, 0.5) {};
		\node [style=none] (11) at (1.5, 0.5) {};
		\node [style=none] (13) at (-2, -0.25) {$X$};
		\node [style=none] (14) at (2, -0.25) {$Y$};
		\node [style=none] (15) at (-1.5, 1.25) {};
		\node [style=none] (16) at (1.5, 1.25) {};
	\end{pgfonlayer}
	\begin{pgfonlayer}{edgelayer}
		\draw (2.center) to (6.center);
		\draw (5.center) to (4.center);
		\draw (8.center) to (3.center);
		\draw (1.center) to (11.center);
		\draw (15.center) to (16.center);
		\draw [bend left=90, looseness=1.75] (16.center) to (11.center);
		\draw [bend right=90, looseness=2.00] (15.center) to (8.center);
	\end{pgfonlayer}
\end{tikzpicture}}{
    f:S\ox X\arrow S\ox Y \quad \begin{tikzpicture}
	\begin{pgfonlayer}{nodelayer}
		\node [style=box] (0) at (0, 0.125) {$f$};
		\node [style=none] (1) at (0.75, 0.5) {};
		\node [style=none] (2) at (0.75, -0.25) {};
		\node [style=none] (3) at (-0.75, 0.5) {};
		\node [style=none] (4) at (-0.75, -0.25) {};
		\node [style=none] (5) at (-1.5, -0.25) {};
		\node [style=none] (6) at (1.5, -0.25) {};
		\node [style=none] (8) at (-1.5, 0.5) {};
		\node [style=none] (11) at (1.5, 0.5) {};
		\node [style=none] (13) at (-2, -0.25) {$X$};
		\node [style=none] (14) at (2, -0.25) {$Y$};
		\node [style=none] (15) at (2, 0.5) {$S$};
		\node [style=none] (16) at (-2, 0.5) {$S$};
	\end{pgfonlayer}
	\begin{pgfonlayer}{edgelayer}
		\draw (2.center) to (6.center);
		\draw (5.center) to (4.center);
		\draw (8.center) to (3.center);
		\draw (1.center) to (11.center);
	\end{pgfonlayer}
\end{tikzpicture}
  }
\end{displaymath}
Interpreted as string diagrams, the equational axioms of the trace operator
capture intuitively equivalent diagrams involving feedback loops.  Two
characteristic axioms are {\em yanking} (left) and {\em dinaturality} (right):
\begin{center}
\begin{tikzpicture}
	\begin{pgfonlayer}{nodelayer}
		\node [style=none] (1) at (-1.5, 0.25) {};
		\node [style=none] (2) at (-1.5, -0.5) {};
		\node [style=none] (3) at (-2.25, 0.25) {};
		\node [style=none] (4) at (-2.25, -0.5) {};
		\node [style=none] (5) at (-3, -0.5) {};
		\node [style=none] (6) at (-0.75, -0.5) {};
		\node [style=none] (8) at (-2.5, 0.25) {};
		\node [style=none] (11) at (-1.25, 0.25) {};
		\node [style=none] (15) at (-2.5, 1) {};
		\node [style=none] (16) at (-1.25, 1) {};
		\node [style=none] (17) at (0, 0.25) {$=$};
		\node [style=none] (18) at (0.75, 0.25) {};
		\node [style=none] (19) at (3, 0.25) {};
	\end{pgfonlayer}
	\begin{pgfonlayer}{edgelayer}
		\draw (2.center) to (6.center);
		\draw (5.center) to (4.center);
		\draw (8.center) to (3.center);
		\draw (1.center) to (11.center);
		\draw (15.center) to (16.center);
		\draw [bend left=90, looseness=1.75] (16.center) to (11.center);
		\draw [bend left=270, looseness=1.75] (15.center) to (8.center);
		\draw (4.center) to (1.center);
		\draw (3.center) to (2.center);
		\draw (18.center) to (19.center);
	\end{pgfonlayer}
\end{tikzpicture}
\begin{tikzpicture}
	\begin{pgfonlayer}{nodelayer}
		\node [style=box] (0) at (0, 0.125) {$f$};
		\node [style=none] (1) at (0.75, 0.5) {};
		\node [style=none] (2) at (0.75, -0.25) {};
		\node [style=none] (3) at (-2, 0.5) {};
		\node [style=none] (4) at (-0.75, -0.25) {};
		\node [style=none] (5) at (-2.75, -0.25) {};
		\node [style=none] (6) at (1.5, -0.25) {};
		\node [style=none] (8) at (-2.5, 0.5) {};
		\node [style=none] (11) at (1.25, 0.5) {};
		\node [style=none] (15) at (-2.5, 1.25) {};
		\node [style=none] (16) at (1.25, 1.25) {};
		\node [style=smallbox] (19) at (-1.75, 0.5) {$g$};
		\node [style=none] (20) at (-1.5, 0.5) {};
		\node [style=none] (21) at (-0.75, 0.5) {};
		\node [style=box] (22) at (4.5, 0.125) {$f$};
		\node [style=none] (23) at (6.5, 0.5) {};
		\node [style=none] (24) at (3, -0.25) {};
		\node [style=none] (25) at (3.75, 0.5) {};
		\node [style=none] (26) at (7.25, -0.25) {};
		\node [style=none] (27) at (5.25, -0.25) {};
		\node [style=none] (28) at (3.75, -0.25) {};
		\node [style=none] (29) at (3.25, 0.5) {};
		\node [style=none] (30) at (7, 0.5) {};
		\node [style=none] (31) at (3.25, 1.25) {};
		\node [style=none] (32) at (7, 1.25) {};
		\node [style=smallbox] (33) at (6.25, 0.5) {$g$};
		\node [style=none] (34) at (5.25, 0.5) {};
		\node [style=none] (35) at (6, 0.5) {};
		\node [style=none] (36) at (2.25, 0.5) {$=$};
	\end{pgfonlayer}
	\begin{pgfonlayer}{edgelayer}
		\draw (2.center) to (6.center);
		\draw (5.center) to (4.center);
		\draw (8.center) to (3.center);
		\draw (1.center) to (11.center);
		\draw (15.center) to (16.center);
		\draw [bend left=90, looseness=2.00] (16.center) to (11.center);
		\draw [bend right=90, looseness=1.75] (15.center) to (8.center);
		\draw (20.center) to (21.center);
		\draw (24.center) to (28.center);
		\draw (27.center) to (26.center);
		\draw (29.center) to (25.center);
		\draw (23.center) to (30.center);
		\draw (31.center) to (32.center);
		\draw [bend left=90, looseness=1.75] (32.center) to (30.center);
		\draw [bend right=90, looseness=1.75] (31.center) to (29.center);
		\draw (34.center) to (35.center);
	\end{pgfonlayer}
\end{tikzpicture}
\end{center}

We will show the delayed trace operator, found in $\caus\CC$, satisfies the
trace operator axioms except yanking and dinaturality. In fact, the delayed
trace of the symmetry yields the morphism that acts as a {\em delay gate}.
Therefore the delayed trace (as its name suggests) may be naturally regarded
as an operation that forms a feedback loop {\em and} inserts the delay gate in
the loop path, depicted as follows:

\ctikzfig{simplefeedback2}

The half-round node is the delay gate, and is filled with its initial state
$p$. The delayed trace operator echoes a principle of synchronous circuit
design: ``all feedback loops should contain a register''.

Our first step towards a delayed trace operator on $\caus\CC$ is to introduce
an operation on 2-cells that converts parts of the value types into the state
space of a computation step.
\begin{definition}\label{def:conversion}
  Let $\cell f S {T\times X}{S'}{T'\times Y}$ be a 2-cell in
  $\sq{\CC}$.  The \emph{value-to-state conversion of $f$ at
    $(T, T')$} is another 2-cell, denoted $\valtostate{f}[T][T']$,
  with the same underlying morphism but different source and target
  1-cells:
  $\cell{\valtostate{f}[T][T']}{S\times T}{X}{S'\times T'}{Y}$.

  When the objects $(T, T')$ involved in the conversion are clear from
  context, we drop them from the notation and write $\valtostate{f}$
  for $\valtostate{f}[T][T']$.
\end{definition}

The value-to-state conversion is depicted inside the tile:

\begin{tikzpicture}
	\begin{pgfonlayer}{nodelayer}
		\node [style=none] (0) at (-1, 2) {};
		\node [style=none] (1) at (1, 2) {};
		\node [style=none] (2) at (-1, 0) {};
		\node [style=none] (3) at (1, 0) {};
		\node [style=smallbox] (4) at (0, 1) {$\phi$};
		\node [style=none] (5) at (0, 2.5) { $S$};
		\node [style=none] (6) at (0, -0.5) { $S'$};
		\node [style=none] (7) at (-1, 0.75) {};
		\node [style=none] (8) at (-0.45, 0.75) {};
		\node [style=none] (9) at (0.45, 0.75) {};
		\node [style=none] (10) at (1, 0.75) {};
		\node [style=none] (11) at (0, 2) {};
		\node [style=none] (12) at (-0.45, 1.25) {};
		\node [style=none] (13) at (0.45, 1.25) {};
		\node [style=none] (14) at (0, 0) {};
		\node [style=none] (15) at (-1.5, 0.25) { $X$};
		\node [style=none] (16) at (1.5, 0.25) { $Y$};
		\node [style=none] (17) at (-1, 1) {};
		\node [style=none] (18) at (-0.45, 1) {};
		\node [style=none] (19) at (0.45, 1) {};
		\node [style=none] (20) at (1, 1) {};
		\node [style=none] (21) at (-1.5, 1.5) { $T$};
		\node [style=none] (22) at (-1.5, 0.875) { $\times$};
		\node [style=none] (23) at (1.5, 0.875) { $\times$};
		\node [style=none] (24) at (1.5, 1.5) { $T'$};
		\node [style=none] (25) at (7.75, 2) {};
		\node [style=none] (26) at (9.75, 2) {};
		\node [style=none] (27) at (7.75, 0) {};
		\node [style=none] (28) at (9.75, 0) {};
		\node [style=smallbox] (29) at (8.75, 1) {$\phi$};
		\node [style=none] (30) at (8.75, 2.5) { $S\times T$};
		\node [style=none] (31) at (8.75, -0.5) { $S'\times T'$};
		\node [style=none] (32) at (7.75, 0.75) {};
		\node [style=none] (33) at (8.3, 0.75) {};
		\node [style=none] (34) at (9.2, 0.75) {};
		\node [style=none] (35) at (9.75, 0.75) {};
		\node [style=none] (36) at (8.5, 2) {};
		\node [style=none] (37) at (8.3, 1.25) {};
		\node [style=none] (38) at (9.2, 1.25) {};
		\node [style=none] (39) at (8.5, 0) {};
		\node [style=none] (40) at (7.25, 1) { $X$};
		\node [style=none] (41) at (10.25, 1) { $Y$};
		\node [style=none] (42) at (9, 2) {};
		\node [style=none] (43) at (8.3, 1) {};
		\node [style=none] (44) at (9.2, 1) {};
		\node [style=none] (45) at (9, 0) {};
		\node [style=none] (46) at (-3, 1) {If $f =$};
		\node [style=none] (47) at (4.25, 1) {, then $\valtostate{f}[T][T'] =$};
		\node [style=none] (48) at (10.75, 1) {.};
	\end{pgfonlayer}
	\begin{pgfonlayer}{edgelayer}
		\draw (0.center) to (1.center);
		\draw (1.center) to (3.center);
		\draw (0.center) to (2.center);
		\draw (2.center) to (3.center);
		\draw (7.center) to (8.center);
		\draw (9.center) to (10.center);
		\draw [in=180, out=-165, looseness=1.75] (11.center) to (12.center);
		\draw [in=15, out=0, looseness=1.25] (13.center) to (14.center);
		\draw (17.center) to (18.center);
		\draw (19.center) to (20.center);
		\draw (25.center) to (26.center);
		\draw (26.center) to (28.center);
		\draw (25.center) to (27.center);
		\draw (27.center) to (28.center);
		\draw (32.center) to (33.center);
		\draw (34.center) to (35.center);
		\draw [in=180, out=-165, looseness=1.75] (36.center) to (37.center);
		\draw [in=15, out=0, looseness=1.25] (38.center) to (39.center);
		\draw [in=180, out=-165, looseness=1.75] (42.center) to (43.center);
		\draw [in=45, out=0, looseness=2.00] (44.center) to (45.center);
	\end{pgfonlayer}
\end{tikzpicture}

The pointwise application of this operation to all the 2-cells in
a stateful morphism sequence is the delayed trace operator. 
\begin{definition}\label{def:dtr}
  Suppose
  $(i,\seq s):\seq T\times\seq X\to \chop{\seq T}\times\seq Y $ is a
  morphism in $\caus\CC$. (Recall
  $\chop{\seq T}=\seqof{\seq T_{\bul + 1}}$.)  The \emph{delayed trace
    of $\seq s$ along $\seq T$ with an initial state
    $p:1\arrow\seq T_0$} is the following morphism in
  $\caus\CC(\seq X, \seq Y)$:
  \begin{displaymath}
    \tr{\seq T} p {(i, \seq s)}
    \teq
    (\langle i, p\rangle,\seqof{\valtostate{\seq s_\bul}[\seq T_\bul][\seq T_{\bul+1}]}).
  \end{displaymath}
\end{definition}

Note this operation is well-defined on extensional equivalence classes of
stateful morphism sequences, and therefore is an operation on $\caus\CC$
morphisms. The delayed trace of $\caus\CC$ already differs from the standard
monoidal trace in two ways: first, the domain and codomain types that are
bound ($\seq T$ and $\chop\seq T$) do not match, and second, the delayed trace
also requires the specification of a global element $p$ called the initial
state. Despite these differences, many of the trace axioms holds for the
delayed trace operator.

\begin{proposition}\label{prop:dtr_props}
  Suppose
  $(i,\seq s): \seq T\times \seq X\arrow \chop{\seq T}\times \seq Y$
  is a morphism in $\caus\CC$. Suppose $(h,\seq r): \seq Y\to\seq Z$,
  $(j,\seq t): \seq W\to\seq X$ and $(f,\seq p): \seq W\to\seq Z$ are
  other arbitrary morphisms in $\caus\CC$. Five standard axioms of 
  monoidal trace, presented in Figure \ref{fig:trace}, hold of delayed trace.
  \begin{figure*}[htbp]
    \begin{displaymath}
      \begin{array}{lrcl}
        \text{Target naturality}
        & (h,\seq r) \circ \tr{\seq T}{p}(i,\seq s) & = & \tr{\seq T}{p}((\id_{\chop\seq T} \times (h,\seq r)) \circ (i,\seq s)) \\
        \text{Source naturality}
        & \tr{\seq T}p(i,\seq s)\circ (j,\seq t) & = & \tr{\seq T}{p}((i,\seq s)\circ (\id_{\seq T} \times (j,\seq t))) \\
        \text{Superposing}
        & \tr{\seq T}{q}(i,\seq s)\times (f,\seq p)  & = & \tr{\seq T}{q}((i,\seq s)\times (f,\seq p)) \\
        \text{Vanishing 1}
        & \tr{\seqof 1}{\id_1}(i,\seq s) & = & (i,\seq s) \\
        \text{Vanishing $\times$}
        & \tr{\seq V}q(\tr{\seq U}p(i,\seq s)) & =  & \tr{\seqof{\seq V_\bul\times\seq U_\bul}}{\langle q,p\rangle}(i,\seq s)
      \end{array}
    \end{displaymath}
    \caption{Equalities Satisfied by Delayed Trace Operator}
    \label{fig:trace}
  \end{figure*}
\end{proposition}

The yanking axiom of the trace operator fails for the delayed trace operator.
Consider the symmetry morphism \mbox{$\sigma_{\seq X,\chop\seq X}:\seq
X\times\chop\seq X\arrow\chop\seq X\times\seq X$} in $\caus\CC$. Define its
delayed trace with an initial state $i:1\arrow\seq X_0$ to be 
$$r_{\seq X}(i)\teq\tr{\seq X}{i}(\sigma_{\seq X,\chop\seq X}):\chop\seq X\arrow\seq X$$

To get a better understanding of $r_{\seq X}(i)$, we first draw
the value-to-state conversion in a single 2-cell in this morphism.

\ctikzfig{yankingsquare}

Doing value-to-state conversion along the whole sequence $\sigma_{\seq X,
\chop\seq X}$ and supplying the initial value $i: 1 \to \seq X_0$ yields:

\ctikzfig{yankingsequence}

We can see that the input at clock tick $k$ is output at clock tick $k+1$.
Therefore, instead of the identity, which is what $r_{\seq X}(i)$ would be if
the yanking axiom held, we have a morphism that operates as a {\em delay
gate}.

The dinaturality axiom of the trace operator also fails for the delayed trace
operator. Dinaturality corresponds to sliding circuits from one end of a
feedback loop to the other, but doing so with a delay gate in the loop affects
the gate's initial state. In digital circuit design, this kind of operation is
called {\em retiming} \cite{Leiserson1991}, and there initial states of
registers is a delicate issue. The delayed trace operator satisfies the
following modified dinaturality property:
\begin{theorem}\label{thm:dinaturality}
  Suppose
  $(i,\seq s):\seq T\times\seq X\arrow\chop{\seq U}\times\seq Y$
  and $(j,\seq g):\seq U\arrow \seq T$ are morphisms in
  $\caus\CC$. For any $u:1\arrow \seq U_0$,
  \begin{align*}
    &\tr{\seq U}{u}(((j,\seq g)\times\seq X)\circ(i,\seq s))\\
    &=\tr{\seq T}{u'}((i,\seq s)\circ((j',\chop\seq g)\times\seq Y))
  \end{align*}
  where $\tuple{j', u'} = U\seq g_0\circ\langle j,u\rangle$.
\end{theorem}
A special case of this modified dinaturality is an abstract version of circuit
retiming, which allows us to commute properly initialized delay gates and
stateless morphisms.
\begin{corollary}
  For any $\seq f:\seq X\arrow\seq Y$ in $\CC^\Nat$, and initial state
  $i:1\arrow\seq X_0$, we have
  $H\seq f\circ r_{\seq X}(i)=r_{\seq Y}(\seq f_0\circ i)\circ H(\chop \seq f)$.
\end{corollary}

The following representation result says that every morphism in $\caus\CC$ can
be obtained as the delayed trace of a stateless morphism.
\begin{theorem}\label{thm:singleregsuffices}
  For any morphism $(i,\seq s)$ in $\caus\CC$, the following equality
  holds:
  \begin{displaymath}
    (i,\seq s)=\tr {\st(i,\seq s)} i (H\seqof{U(\seq s_\bul)})
  \end{displaymath}
\end{theorem}
This theorem is our formalization of folklore knowledge that every synchronous
digital circuit can be written as a single combinational (stateless) circuit
plus a feedback loop with a register.

\subsection{Delayed Trace in $\causz\CC$}

The category $\causz\CC$ is also closed under the delayed trace
operator. Since $\seq X = \chop \seq X$ in $\causz\CC$, delayed 
dinaturality is even closer to true dinaturality.
\begin{corollary}\label{cor:regular_dinaturality}
  Suppose
  $(i,\seqof s):\seqof T\times\seqof X\arrow\seqof U\times\seqof Y$ is
  a morphism in $\causz\CC$, and $(j,\seqof g):\seqof U\arrow\seqof T$
  is another morphism in $\causz\CC$. For any initial state $u:1\arrow U$,
  \begin{align*}
    &\tr{\seqof U}{u}(((j,\seqof g)\times\id_{\seqof X})\circ(i,\seqof s))\\
    &=\tr{\seqof T}{u'}((i,\seqof s)\circ((j',\seqof g)\times\id_{\seqof Y}))
  \end{align*}
  where $\tuple{j', u'} = g\circ\langle j,u\rangle$.
\end{corollary}
\begin{corollary}
  For any $f:X\arrow Y$ in $\CC$, and initial state $i:1\arrow X$, we
  have $H\seqof f\circ r_{[X]}(i)=r_{[Y]}(f\circ i)\circ H(\seqof f)$.
\end{corollary}

\subsection{Diagrammatic reasoning about $\causz\CC$ morphisms}

Here we informally introduce a diagrammatic syntax for morphisms in
$\causz\CC$. Theorem~\ref{thm:singleregsuffices} indicates that we can
generate all $\causz\CC$ morphisms with the following grammar:
$$\varphi ::= H_0f | \varphi_1\circ\varphi_2 | \varphi_1\times\varphi_2 | \tr{\seq S}{i}(\varphi)$$
\noindent where $f$ is a $\CC$-morphism. We generate circuit
diagrams with a parallel 2-dimensional grammar:
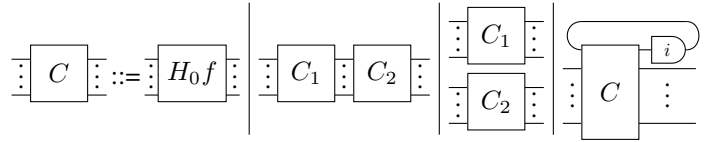
\ctikzfig{circuitgrammar}
where the box labeled $f$ has $m$ inputs and $n$ outputs when \mbox{$f: \prod_{k = 1}^m A_k
\arrow \prod_{k=1}^n B_k$}. As is typical in string diagrams,
$H_0\id_X$ is depicted by a wire and $H_0\sigma_{X,Y}$ by a wire crossing.
Additionally, we depict $H_0!_A$ and $H_0\tuple{\id_A,\id_A}$ with a discarder
and copier: \begin{tikzpicture}
	\begin{pgfonlayer}{nodelayer}
		\node [style=none] (0) at (-0.375, 0) {};
		\node [style=black dot] (1) at (0.25, 0) {};
		\node [style=none] (2) at (2, 0) {};
		\node [style=black dot] (3) at (2.5, 0) {};
		\node [style=none] (4) at (3, 0.5) {};
		\node [style=none] (5) at (3, -0.5) {};
		\node [style=none] (6) at (1.25, 0) {and};
	\end{pgfonlayer}
	\begin{pgfonlayer}{edgelayer}
		\draw (0.center) to (1);
		\draw (2.center) to (3);
		\draw [bend right=45] (4.center) to (3);
		\draw [bend right=45] (3) to (5.center);
	\end{pgfonlayer}
\end{tikzpicture}.

The evident interpretation in $\causz\CC$ of these diagrams induces an
equivalence on such diagrams. For instance, as a special case of
corollary~\ref{cor:regular_dinaturality}, sliding a stateless node along a
loop is possible by changing the value in the delay gate:
\ctikzfig{slideandretime}

As an example of diagrammatic reasoning, we show that this simple delayed
dinaturality plus superposing allows us to obtained delayed dinaturality for
stateful circuits (\cref{thm:dinaturality}).
\ctikzfig{superslide}

More formal treatment of this diagrammatic equational system can be done
through the construction of the free cartesian category with the delayed trace
operator. We reserve this formal axiomatization for future work, and move on
to the study of the differentiablity of the causal computations realized by
$\caus\CC$.




\section{Cartesian Differential Structure}\label{sec:seqcdc}

In this section, we investigate differentiation in $\caus\CC$. Our primary
tool is the theory of Cartesian differential categories, introduced by Blute,
Cockett, and Seely in~\cite{cartesiandiffcat}. We begin by recalling background.

\begin{definition}[\cite{cartesiandiffcat}]
  A \emph{left additive category} is a Cartesian category such that
  every object has a designated commutative monoid structure, which we
  write $+_X: X\times X \to X$ and $0_X: 1 \to X$.  These commutative
  monoids must be compatible with the Cartesian structure of the
  category by satisfying:
  \begin{align*}
    0_{X\times Y} &= 0_X \times 0_Y \\
    +_{X\times Y} &= (+_X \times +_Y)\circ(X\times \sigma_{Y,X}\times Y)
  \end{align*}
\end{definition}

The vector space structure on Euclidean spaces is a classic example of
left additive structure.

\begin{example}[\cite{cartesiandiffcat}]\label{ex:classical_leftadd}
  The category $\Euci$ whose objects are $\RR^n$ for $n \geq 0$
  and morphisms are smooth functions is a left additive category,
  where $+_{\RR^n}$ is the sum of vectors in $\RR^n$ and $0_{\RR^n}$
  is the zero vector in $\RR^n$.
\end{example}

To obtain left additive structure for $\caus\CC$, it suffices to take
sequences of the corresponding pieces of left additive structure for
$\CC$, much like how the Cartesian structure of $\CC$ lifted.

\begin{lemma}\label{lem:leftadd_lifts}
  If $\CC$ is a left additive category, so is $\caus\CC$.
\end{lemma}

Next, we introduce some helpful families of morphisms present in every
Cartesian left additive category that are useful for condensing later
definitions.

\begin{definition}\label{def:helper_morphs}
  Let $\CC$ be a Cartesian left additive category.  For every object
  $X$ from $\CC$ [or pair of objects $(X, Y)$], let
  \begin{itemize}
  \item $\delta_{X,Y} \teq X \times \sigma_{Y, X} \times Y$
  \item $\alpha_X \teq \delta_{X, X} \circ(X \times X\times \Delta_X)$
  \item
    $\beta_X \teq X\times \final{X} \times X\times X$
  \item
    $\gamma_{X,Y} \teq (X\times \Delta_X\times Y \times Y) \circ
    \delta_{X,Y}$
  \item $\zeta_X \teq X\times 0_{X\times X} \times X$
  \end{itemize}
\end{definition}

Now we are ready to describe the central object of our study this
section, Cartesian differential categories.

\begin{definition}\label{def:cdc}
  A \emph{Cartesian differential category} is a left additive category
  $\CC$ with a \emph{Cartesian differential operator}
  \mbox{$D: \CC(X, Y) \to \CC(X\times X, Y)$}, satisfying:
  \begin{enumerate}[CD1.]
  \item $Ds = s\times {\final{\dom (s)}}$ for
    $s \in \{X, \sigma_{X,Y}, \final{X}, \Delta_X, +_X, 0_X\}$
  \item $Df\circ(0_X \times X) = 0_Y \circ \final{X}$
  \item $Df \circ (+_X\times X) = +_Y\circ(Df\times Df)\circ \alpha_X$
  \item $D(g\circ f) = Dg \circ (Df \times f) \circ (X \times \Delta_X)$
  \item $D(f\times h) = (Df \times Dh)\circ\delta_{X,V}$
  \item $DDf\circ\zeta_X = Df$
  \item $DDf\circ\delta_{X,X} = DDf$
  \end{enumerate}
  for all $f: X \to Y$, $g: Y \to Z$, and $h: V \to W$.
\end{definition}

This definition of a Cartesian differential category is not exactly
that of \cite{cartesiandiffcat}, but it is mostly
straightfoward to check that they are equivalent. The biggest changes
are in axioms CD6 and CD7, for which we have taken alternate forms
given in~\cite[Proposition 4.2]{sdg2014}.

\begin{example}[\cite{cartesiandiffcat}] \label{ex:classic_diff}
  $\Euci$ is a Cartesian differential category. The differential operator $D$
  sends a smooth function $f: \RR^n \to \RR^m$ to $Df: (x_1, x_2) \mapsto
  Jf|_{x_2}\times x_1$, where $Jf|_{x_2}$ is the Jacobian matrix of $f$ evaluated
  at $x_2$.
\end{example}

In light of the standard example, we can describe the ideas behind the CD
axioms. CD1 says that the basic morphisms provided by the structure of the
Cartesian left additive category are linear (in the sense that
$Js|_{x_2}\times x_1 = s(x_1)$), while CD2 and CD3 express the fact that
$Jf|_{x_2}\times x_1$ is linear (in the sense of linear algebra) in its $x_1$
argument. CD4 is the chain rule, while CD5 says the derivative of a parallel
composition is the parallel composition of derivatives. CD6 and CD7 have to do
with partial derivatives: CD7 is the symmetry of partial derivatives, and CD6
is trickier to describe exactly, but is related to the linearity of partial
derivatives.

Many of the CD axioms mention the parallel composition of morphisms with
$\times$. When we state these in $\caus\CC$, it will be helpful to have an
operation for forming parallel compositions. This motivates us
to define the following operation on 2-cells.

\begin{definition}\label{def:cc}
  Let $\cell f S X {S'} Y$ and $\cell k T Z {T'} W$ be arbitrary
  2-cells from $\sq\CC$. The \emph{cross composition} of $f$ and $k$
  is another 2-cell
  $\cell{f\cc k}{S\times T}{X\times Z}{S'\times T'}{Y\times W}$
  defined by
  $$f\cc k = (\valtostate{(T\times Y)^h}[T][T]\vc k)\hc(f\vc \valtostate{(S'\times Z)^h}[S'][S']).$$
\end{definition}

It may be easier to understand $\cc$ composition by its underlying
morphism:

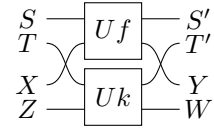
\ctikzfig{crosscomp}

The idea of this operation is to execute two 2-cells in parallel,
without their states or values interacting with each other. We are
purposefully avoiding using $\times$ for $\cc$ so as not to imply
there is some kind of Cartesian structure on the double category
$\sq\CC$.

To avoid using too many grouping symbols when disambiguating 2-cell
expressions involving $\vc$, $\hc$, and $\cc$ we will say $\cc$ binds
tightest, then $\vc$, and last $\hc$, so $f\vc g \cc h \hc k$ means
$(f\vc (g \cc h)) \hc k$.

As desired, this operation implements Cartesian product in $\caus\CC$.
\begin{lemma}
  $(i, \seq f)\times (j, \seq g) = (\tuple{i, j}, \seqof{\seqel f \cc \seqel g})$
  for all $\caus\CC$ morphisms $(i, \seq f)$ and $(j, \seq g)$.
\end{lemma}

We can now start defining the Cartesian differential operator on $\caus\CC$.
For the remainder of this section we assume $\CC$ is a Cartesian differential
category and let $D$ be its differential operator. We start by defining our
differential operator within a time step, by giving some operations on
2-cells.

\begin{definition}\label{def:squD}
  We define two endofunctions on 2-cells from $\sq{\CC}$.  The first,
  $\squD_0$, takes the 2-cell $\cell f S X {S'} Y$ to the 2-cell
  $\cell{\squD_0 f}{S\times S}{X\times X}{S'}{Y}$ with
  $U\squD_0 f \triangleq DUf \circ \delta_{S, X}$.

  The second, $\squD$, takes $f$ to
  $\cell{\squD f}{S\times S}{X\times X}{S'\times S'}{Y}$ with
  $$\squD f \teq (S\times \Delta_S)^v; (\squD_0 f \cc (\final{Y}^h *f)) * (X\times \Delta_X)^h.$$
\end{definition}

The string diagrams for the underlying morphisms of $\squD_0$ and
$\squD$ may be easier to understand. For $\squD_0$,

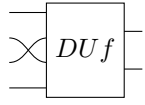
\ctikzfig{simplediff}

\noindent while for $\squD f$,

\begin{tikzpicture}
	\begin{pgfonlayer}{nodelayer}
		\node [style=tallbox] (0) at (2.5, 0) {$DUf$};
		\node [style=none] (1) at (1.5, 1) {};
		\node [style=none] (2) at (1.5, 0.325) {};
		\node [style=none] (3) at (1.5, -0.325) {};
		\node [style=none] (4) at (1.5, -1) {};
		\node [style=none] (5) at (-2.5, 1) {};
		\node [style=none] (6) at (0.5, 0.325) {};
		\node [style=none] (7) at (0.5, -0.325) {};
		\node [style=none] (8) at (0.5, -1) {};
		\node [style=none] (9) at (3.5, -0.5) {};
		\node [style=none] (10) at (4, -0.5) {};
		\node [style=none] (11) at (3.5, 0.5) {};
		\node [style=none] (12) at (5.25, 0.5) {};
		\node [style=box] (13) at (2.5, -2.75) {$Uf$};
		\node [style=none] (14) at (1.75, -2.25) {};
		\node [style=none] (15) at (0.5, -2.25) {};
		\node [style=none] (18) at (3.25, -3.25) {};
		\node [style=black dot] (19) at (4, -3.25) {};
		\node [style=none] (20) at (3.25, -2.25) {};
		\node [style=none] (21) at (4, -2.25) {};
		\node [style=none] (22) at (5.25, -0.5) {};
		\node [style=none] (23) at (5.25, -2.25) {};
		\node [style=none] (24) at (-1, -2.575) {};
		\node [style=none] (25) at (-1, -1.9) {};
		\node [style=none] (26) at (-1, -0.325) {};
		\node [style=none] (28) at (-1, 0.325) {};
		\node [style=none] (29) at (1.75, -3.25) {};
		\node [style=black dot] (30) at (-1.75, -3.025) {};
		\node [style=black dot] (31) at (-1.75, 0) {};
		\node [style=none] (33) at (-2.5, 0) {};
		\node [style=none] (34) at (-2.5, -3) {};
		\node [style=none] (35) at (-2.5, -1.9) {};
	\end{pgfonlayer}
	\begin{pgfonlayer}{edgelayer}
		\draw [style=state] (5.center) to (1.center);
		\draw [style=state, in=-180, out=0] (6.center) to (3.center);
		\draw [in=180, out=0] (7.center) to (2.center);
		\draw (8.center) to (4.center);
		\draw (9.center) to (10.center);
		\draw [style=state] (11.center) to (12.center);
		\draw [style=state] (15.center) to (14.center);
		\draw (18.center) to (19);
		\draw [style=state] (20.center) to (21.center);
		\draw [style=state, in=180, out=0, looseness=0.50] (21.center) to (22.center);
		\draw [style=state] (28.center) to (6.center);
		\draw [style=state, in=0, out=180, looseness=0.50] (15.center) to (26.center);
		\draw [in=0, out=-180, looseness=0.50] (7.center) to (25.center);
		\draw [in=0, out=-180, looseness=0.50] (8.center) to (24.center);
		\draw [in=180, out=0, looseness=0.50] (10.center) to (23.center);
		\draw [style=state] (33.center) to (31);
		\draw [style=state, in=-180, out=90] (31) to (28.center);
		\draw [style=state, in=-180, out=-90] (31) to (26.center);
		\draw (35.center) to (25.center);
		\draw (34.center) to (30);
		\draw [in=180, out=90, looseness=1.25] (30) to (24.center);
		\draw [in=-180, out=-90, looseness=0.50] (30) to (29.center);
	\end{pgfonlayer}
\end{tikzpicture} = 
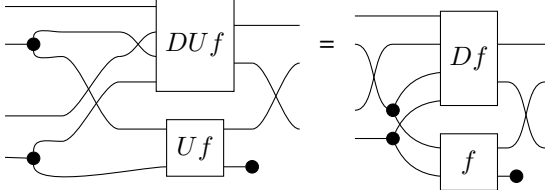
\begin{tikzpicture}
	\begin{pgfonlayer}{nodelayer}
		\node [style=tallbox] (25) at (0.5, -0.375) {$Df$};
		\node [style=none] (26) at (1.25, -1) {};
		\node [style=none] (27) at (1.25, 0) {};
		\node [style=none] (28) at (-0.25, -1.5) {};
		\node [style=none] (29) at (-0.25, -0.75) {};
		\node [style=none] (30) at (-0.25, 0) {};
		\node [style=none] (31) at (-0.25, 0.75) {};
		\node [style=box] (32) at (0.5, -3.125) {$f$};
		\node [style=none] (33) at (1.25, -2.75) {};
		\node [style=none] (34) at (1.25, -3.5) {};
		\node [style=none] (35) at (-0.25, -2.75) {};
		\node [style=none] (36) at (-0.25, -3.5) {};
		\node [style=black dot] (37) at (-1.5, -1.75) {};
		\node [style=black dot] (38) at (-1.5, -2.5) {};
		\node [style=none] (39) at (-2.5, -1.75) {};
		\node [style=none] (40) at (-2.5, -2.5) {};
		\node [style=none] (42) at (-1.5, 0) {};
		\node [style=none] (43) at (-2.5, 0.75) {};
		\node [style=none] (44) at (-2.5, 0) {};
		\node [style=none] (45) at (1.5, -1) {};
		\node [style=none] (47) at (2.5, -1) {};
		\node [style=none] (48) at (2.5, 0) {};
		\node [style=none] (49) at (1.5, -2.75) {};
		\node [style=none] (51) at (2.5, -2.75) {};
		\node [style=black dot] (52) at (1.75, -3.5) {};
		\node [style=none] (57) at (-1.5, -1.75) {};
	\end{pgfonlayer}
	\begin{pgfonlayer}{edgelayer}
		\draw [style=state] (43.center) to (31.center);
		\draw [style=state, bend left] (37) to (29.center);
		\draw [style=state, bend right] (37) to (35.center);
		\draw [style=state] (27.center) to (48.center);
		\draw [style=state] (33.center) to (49.center);
		\draw [style=state, in=-180, out=0, looseness=0.75] (49.center) to (47.center);
		\draw [in=180, out=0, looseness=0.75] (39.center) to (42.center);
		\draw (42.center) to (30.center);
		\draw (40.center) to (38);
		\draw [bend left] (38) to (28.center);
		\draw [bend right] (38) to (36.center);
		\draw (34.center) to (52);
		\draw [in=0, out=-180, looseness=0.75] (51.center) to (45.center);
		\draw (45.center) to (26.center);
		\draw [style=state, in=180, out=0, looseness=0.75] (44.center) to (57.center);
	\end{pgfonlayer}
\end{tikzpicture}

The Cartesian differential operator on $\caus\CC$ is based on $\squD$,
and so to prove that it is a differential operator, we need some
properties of $\squD$.

\begin{proposition}\label{prop:squD_properties}
  Let $\cell f S X {S'} Y$, $\cell g T Y {T'} Z$,
  $\cell h {S'} Z {S''} W$, and $\cell k T Z {T'} W$ be arbitrary
  2-cells. The following are properties of $\squD$:
  \begin{enumerate}
  \item If $\varphi \in \CC(X, Y)$, then
    \mbox{$\squD(\varphi^h) = (D\varphi)^h$} and
    \mbox{$\squD(\varphi^v) = ((D\varphi\times \varphi)\circ(X\times
      \Delta_X))^v$}.
  \item
    $(0_S\times S)^v\vc \squD f \hc (0_X \times X)^h = (0_Y \circ
    !_Y)^h \hc f\vc (0_{S'} \times S')^v$
  \item $(+_S\times S)^v\vc \squD f \hc (+_X \times X)^h \\
      = \alpha_S^v\vc(+_Y^h \hc \squD f \cc \squD f \hc \alpha_X^h)\vc \beta_{S'}^v \vc (+_{S'}\times S')^v$
  \item $\squD(f\vc h) = (\squD f\vc \squD h) \hc \delta_{X,Z}^h$
  \item
    $\squD(g\hc f)\vc \gamma_{S', T'}^v = \gamma_{S, T}^v\vc (\squD g
    \hc (\squD f \cc f)\hc (X\times \Delta_X)^h)$
  \item
    $\squD(f \cc k)\vc \delta_{S', T'}^v = \delta_{S, T}^v\vc \squD
    f \cc \squD k \hc \delta_{X, Z}^h$
  \item
    $\zeta_S^v\vc \squD \squD f \hc \zeta_X^h = \squD f\vc
    \zeta_{S'}^v$
  \item
    $\delta_{S,S}^v\vc \squD \squD f \hc \delta_{X, X}^h = \squD \squD
    f\vc \delta_{S', S'}^h$
  \end{enumerate}
\end{proposition}

The method to prove these properties is conceptually simple: use the
definitions of the operations on 2-cells (and properties of left additive
categories and CD axioms) to check that both sides of each equation have the
same boundary 1-cells and the same underlying $\CC$-morphism. Practically, the
underlying morphisms are complex, so this turns into an intense string diagram
exercise, which can be found in the appendix.

An important consequence of Proposition~\ref{prop:squD_properties}(4)
is the following extension to finite sequences of vertically composed
2-cells.

\begin{lemma}\label{lem:squD_vertical_composition}
  Let $(f_k)_{k=0}^n$ be a finite sequence of vertically composable
  2-cells. Then
  $\squD(f_0 \vc\cdots\vc f_n) = (\squD f_0 \vc\cdots\vc\squD f_n)\hc
  z^h$, where $z$ is the unzipping isomorphism in $\CC$ of type
  \begin{displaymath}
    \textstyle
    \prod_{k=0}^n(\dom f_k\times \dom f_k)
    \arrow
    (\prod_{k=0}^n\dom f_k)\times(\prod_{k=0}^n\dom f_k).
  \end{displaymath}
\end{lemma}

We can now state the operator we seek on $\caus\CC$.

\begin{definition}
  The componentwise application of $\squD$ to 2-cells in a $\caus\CC$
  morphism,
  $\seqD: (i, \seq s) \mapsto (U\squD (i^v), \seqof{\squD {\seq
      s}_\bul})$, is a well-defined operation on $\caus\CC$ morphisms
  of type
  $$\seqD: \caus\CC(\seq X, \seq Y) \to \caus\CC(\seq X\times\seq X, \seq Y).$$
\end{definition}

A key contribution of this work is the fact that this operation is
actually a Cartesian differential operator.

\begin{theorem}\label{pp:cartdiff}
  $\seqD$ is a Cartesian differential operator.
\end{theorem}

The strategy for this proof is to use the properties of $\squD$ from
Proposition~\ref{prop:squD_properties}, which were selected to be used with
the Shim Lemma to obtain the CD axioms. For example, in this context, CD4 (the
chain rule) states: 
$$\seqD ((j, \seq g)\circ(i, \seq f)) =
\seqD(j, \seq g)\circ(\seqD(i, \seq f)\times (i, \seq f))\circ(\seq X
\times \Delta_{\seq X}).$$

The key step in proving this is invoking the Shim Lemma with $\seqel b = \gamma_{\seqel S, \seqel T}$.
We have two conditions to check for this invocation:
\mbox{$\gamma_{\seq S_0, \seq T_0}\circ\tuple{0_{\seq S_0}, 0_{\seq T_0}, i, j} = 
\tuple{0_{\seq S_0}, i, i, 0_{\seq T_0}, j}$}
and
\begin{align*}
\squD(&\seqel g\hc \seqel f)\vc \gamma_{\nextseqel S,\nextseqel T}^v \\
&= \gamma_{\seqel S,\seqel T}^v\vc (\squD \seqel g \hc (\squD \seqel f \cc \seqel f)\hc (\seqel X\times \Delta_{\seqel X})^h),
\end{align*}
\noindent the latter of which is a case of Proposition~\ref{prop:squD_properties}(5). 

We can now prove CD4 for $\seqD$:
\begin{align*}
\seqD(&(j, \seq g) \circ (i, \seq f)) = (\tuple{0_{\seq S_0}, 0_{\seq T_0}, i, j}, \seqof{\squD(\seqel g\hc \seqel f)}) \\
&= (\tuple{0_{\seq S_0}, i, i, 0_{\seq T_0}, j}, 
\seqof{\squD \seqel g \hc (\squD \seqel f \cc \seqel f)\hc (\seqel X\times \Delta_{\seqel X})^h}) \\
&= \seqD(j, \seq g)\circ (\seqD(i, \seq f) \times (i, \seq f)) \circ (\seq X\times \Delta_{\seq X})
\end{align*}   
\noindent where the second line is the Shim Lemma step. 
The other axioms are similar and can be found in the appendix.

The following result demonstrates that our differential operator matches (up
to isomorphism) the unroll-and-differentiate procedure used in backpropagation
through time.


\begin{theorem}\label{th:unrol}
  For any morphism $(i,\seq f):\seq A\arrow\seq B$ in $\caus\CC$,
  \begin{displaymath}
    \Un_k(\seqD(i,\seq f))=D(\Un_k(i,\seq f))\circ z:
    \prod_{n=0}^k(\seq A_n\times\seq A_n)\arrow\seq B_k,
  \end{displaymath}
  where $z$ is the unzipping isomorphism from lemma~\ref{lem:squD_vertical_composition}.
\end{theorem}

\section{Differentiation of causal morphisms}

For our applications, we note that $\seqD$ restricts to $\causz\CC$.

\begin{corollary}\label{cor:cartdiff_causz}
  The operation $\seqD$ restricted to $\causz\CC$ is a Cartesian
  differential operator on $\causz\CC$.
\end{corollary}

Using this differential operator in $\causz\CC$, we can find the derivative of
a stateful function as another stateful function. From the definition of
$\squD$ on 2-cells, we know:

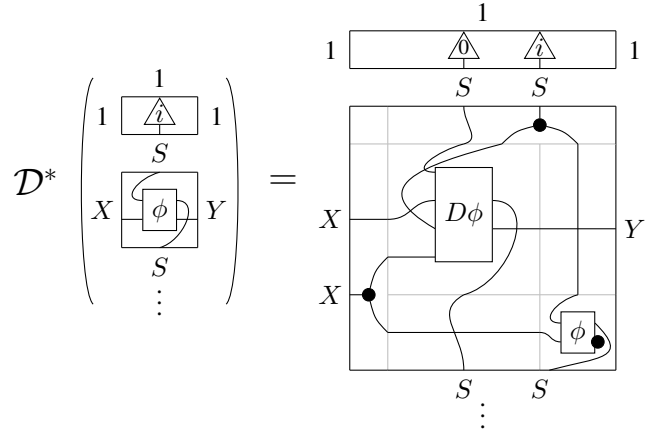
\ctikzfig{diff2cell}

Translating this fact along the correspondence between circuit diagrams and
morphisms in $\causz\CC$, we obtain the following diagram as the derivative of
our simple stateful function. (The red dashed boxes do not have any
mathematical meaning; they are only there so we can describe how the device on
the right works.)

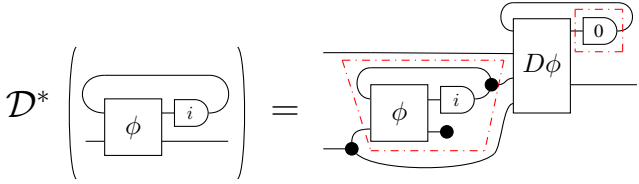
\ctikzfig{diffcircuit}

Again, the idea of a derivative in a Cartesian differential category is to
take a base point $\seq x$ as its lower argument and a small change as
its upper argument $\seq{\Delta x}$ and return an approximation for the
difference between the outputs of the function at $\seq x$ and the function
at $\seq x + \seq{\Delta x}$.

Here is how the device obtained above accomplishes this. The red trapezoidal
region is a copy of the original device which maintains the current state of
the function in the delay gate initialized with $i$. It uses this state itself
to maintain this invariant, and supplies a copy to the derivative of the
combinational part, $D\phi$. Therefore, the bottom two arguments received by
the $D\phi$ subdevice are the state and value inputs $\phi$ would receive.

In the upper delay gate (initialized to $0$, also boxed in red), the device
accumulates its best approximation for the \textbf{difference between states}
between the original device executed at $\seq x$ and at $\seq x + \seq{\Delta
x}$, using the current state and input values, the approximate state change
supplied from the upper delay gate, and the value change supplied at the upper
input (above the red trapezoid). Meanwhile, the output wire to the left
reports the best approximation for the difference in outputs to the
environment.

Though it may seem we have taken a slightly special case by assuming the
$\phi$ device is stateless (being an underlying morphism from a 2-cell),
theorem~\ref{thm:singleregsuffices} ensures \emph{all} $\causz\CC$ morphisms
can be written in this form. So in fact this is a fully abstract circuit
diagram for derivatives in $\causz\CC$.

Taking $\CC = \Euci$, this string diagram specializes to the derivative
of a recurrent neural network. Theorem~\ref{th:unrol} guarantees this
derivative matches precisely what we expect from the unroll-and-differentiate
procedure used in backpropagation through time. However, the extra structure
we have discovered for this procedure, namely that $\caus{\Euci}$ is a
Cartesian differential category, give us many useful properties.
For example, the derivative of

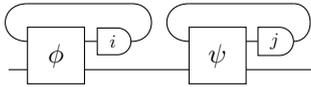
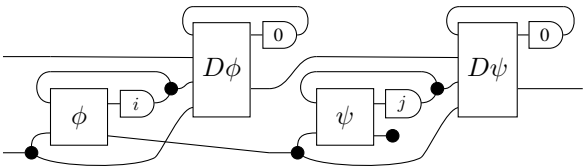
\ctikzfig{composedsimple}

\noindent is

\ctikzfig{composedsimple_deriv}

\noindent We leave this exercise to the reader.

\section{Conclusion and future directions}

We have shown how to treat differentiation of stateful functions via two
pieces of categorical machinery. First, we described the $\caus-$
construction, taking a category and augmenting it with morphisms representing
causal functions of sequences. The special subcategory $\causz-$ consists of
constant stateful morphism sequences, which perform the same computation at
each clock tick, much like a Mealy machine.

Second, we showed that this causal construction also admits an abstract form
of differentiation. Our key technical results showed that if a category $\CC$
is a Cartesian differential category, so is $\caus\CC$. In particular, this
allows us to give a finite representation of the derivative of a causal
function. In addition to being much more compact than the well-known
unroll-then-differentiate approach, the structure of Cartesian differential
categories ensure this differentiation operation has many useful properties of
derivatives of undergraduate calculus, including a chain rule.

We believe that experimentation in machine learning will use differentiation
and gradients in many new and interesting contexts. We also believe the
abstract nature of Cartesian differential categories will prove very valuable
for organizing the theory behind this growing field.

Though we would like to say our abstract treatment of differentiation can be
used directly by machine learning practitioners, it appears this is not the
case yet. The derivative of a morphism in a Cartesian differential category is
not the same as having an explicit Jacobian or gradient. A gradient can be
recovered from this morphism by applying it to all the basis vectors, but when
there are millions of parameters in a machine learning model, this idea is
computationally disastrous. We think that by adding some structure to
Cartesian differential categories, such as a designated closed subcategory, we
could give a theoretical treatment allowing for more explicit representation
of Jacobians.

Another issue with our work is that often machine learning
practitioners often use functions which are not smooth, not
differentiable, and even sometimes partial! While this wrinkle is easy
enough to overcome in practice so long as it is encountered
sufficiently rarely, to theoreticians it can be more of a
challenge. Enhancing this work with differential restriction
categories might be a good way forward.

An interesting observation that points to potential further
applications in machine learning is the following. We know that $\CC$
being Cartesian differential category implies $\causz\CC$ is as
well. Therefore, $\causz{\causz\CC}$ is \textbf{also} a Cartesian
differential category.  Morphisms in $\CC$ process individual inputs
and morphisms in $\causz\CC$ process sequences of inputs, so morphisms
in $\causz{\causz\CC}$ process sequences of sequences of
inputs. Similarly, while $\causz\CC$ adds delay gates whose values can
change as elements of its input sequence are processed,
$\causz{\causz\CC}$ will add \emph{meta-delay gates} whose values
change after a single sequence in its input sequence of sequences has
been processed. This behavior of meta-delay gates seems a lot like
parameter updating after processing an example in the training of
neural networks. Further iterating this construction to
$\causz{\causz{\causz\CC}}$ may be a good way to model a
hyperparameter tuning process.

Somewhat removed from potential machine learning applications, we are
also curious about the further development of the theory of delayed
traces. In particular, it seems there are quite a few interesting
delayed traces besides the one we described for
$\caus\CC$.

\section*{Acknowledgments}

We are very grateful to Bart Jacobs, Fabio Zanasi, Nian-Ze Lee and Masahito
Hasegawa for many useful discussions. Thanks to JS Lemay for the pointer to
\cite{sdg2014}. Thanks also to the developers of TikZiT (Aleks Kissinger,
Alexander Merry, Chris Heunen, K.~Johan Paulsson), which which many of these
diagrams were made.

The authors are supported by ERATO HASUO Metamathematics for Systems
Design Project (No. JPMJER1603), JST.

\bibliographystyle{plain} \bibliography{main}

\onecolumn
\appendix

\section{Appendix}

\subsection{Proof of \cref{causal}}

\newcommand{\tmmathbf}[1]{\ensuremath{\boldsymbol{#1}}}
\newcommand{\tmop}[1]{\ensuremath{\operatorname{#1}}}

\paragraph{Injectivity}

Without loss of generality, we assume $A$ is non-empty and take
$\bot \in A$.  We define helper functions
$(-)_{\leqslant k} : A^{\Nat} \rightarrow A^{k + 1}$ and
$(-)^{k +} : A^{k + 1} \rightarrow A^{\Nat}$ by
\[ a_{\leqslant k} = (a_0, \ldots, a_k), \quad (a_0, \ldots, a_k)^{k
    +} = a_0, \ldots, a_k, [\bot] . \]
\begin{lemma}
  For any causal function $f : A^{\Nat} \rightarrow B^{\Nat}$, we
  have $(f ((a_{\leqslant k})^{k +}))_k = (f (a))_k$.
\end{lemma}

\begin{proof}
  Notice that $(a_{\leqslant k})^{k +} \equiv_k a$. From causality, we
  have $f ((a_{\leqslant k})^{k +}) \equiv_k f (a)$, and hence
  $(f ((a_{\leqslant k})^{k +}))_k = (f (a))_k$.
\end{proof}

Let $f : A^{\Nat} \rightarrow B^{\Nat}$ be a causal function. We
define $s (f) = (\tmop{id}_1, \tmmathbf{f})$ by
\[ \tmmathbf{f}_k : A^k \times A \rightarrow A^{k + 1} \times B \]
\[ \tmmathbf{f}_k (x, a) = ((x, a), (f ((x, a)^{k +}))_k) \]
\begin{lemma}
  For any $k \in \Nat$ and $a \in A^{\Nat}$, we have
  $U (i ; \tmmathbf{f}_0 ; \ldots ; \tmmathbf{f}_k) (a_{\leqslant k})
  = (a_{\leqslant k}, f (a)_{\leqslant k}) .$
\end{lemma}

\begin{proof}
  When $k = 0$,
  \[ \tmop{Tc}_0 (\tmop{id}_1, \tmmathbf{f}) (a_{\leqslant 0}) = (a_0,
    f ((a_0)^{0 +})_0) = (a_0, f (a)_{\leqslant 0}) . \] Suppose that
  $U (i ; \tmmathbf{f}_0 ; \ldots ; \tmmathbf{f}_k) (a_{\leqslant k})
  = (a_{\leqslant k}, f (a)_{\leqslant k})$ holds. Then
  \begin{eqnarray*}
    & & U (i ; \tmmathbf{f}_0 ; \ldots ; \tmmathbf{f}_k ; \tmmathbf{f}_{k + 1})
        (a_{\leqslant k + 1})\\
    & = & \tmop{let} (s', b) = U (i ; \tmmathbf{f}_0 ;
          \ldots ; \tmmathbf{f}_k) (a_{\leqslant k}) \tmop{in}\\
    &  & \tmop{let} (s'', b') = U\tmmathbf{f}_{k + 1} (s', a_{k + 1})
         \tmop{in} (s'', (b, b'))\\
    & = & \tmop{let} (s'', b') = U\tmmathbf{f}_{k + 1} (a_{\leqslant k}, a_{k
          + 1}) \tmop{in} (s'', (f (a)_{\leqslant k}, b'))\\
    & = & (a_{\leqslant k + 1}, (f (a)_{\leqslant k}, f ((a_{\leqslant k +
          1})^+)_{k + 1}))\\
    & = & (a_{\leqslant k + 1}, (f (a)_{\leqslant k}, f (a)_{k + 1}))\\
    & = & (a_{\leqslant k + 1}, f (a)_{\leqslant k + 1}) .
  \end{eqnarray*}
  
\end{proof}

\begin{corollary}
  $\tmop{Tc}_k (\tmop{id}_1, \tmmathbf{f}) (a_{\leqslant k}) = f
  (a)_{\leqslant k}$.
\end{corollary}

Suppose that $s (f)$ and $s (g)$ are extentionally equivalent. That
is,
$\tmop{Tc}_{\bullet} (\tmop{id}_1, \tmmathbf{f}) = \tmop{Tc}_{\bullet}
(\tmop{id}_1, \tmmathbf{g})$. Then for any $a \in A^{\Nat}$, we have
\[ f (a)_{\leqslant \bullet} = \tmop{Tc}_{\bullet} (\tmop{id}_1,
  \tmmathbf{f}) (a_{\leqslant \bullet}) = \tmop{Tc}_k (\tmop{id}_1,
  \tmmathbf{g}) (a_{\leqslant \bullet}) = g (a)_{\leqslant \bullet}
  . \] Therefore $f = g$.

\paragraph{Surjectivity}

Let $(i, \tmmathbf{f}) : [A] \rightarrow [B]$ be a stateful morphism
sequence.  We define a function
$g : A^{\Nat} \rightarrow B^{\Nat}$ by
\[ g (a)_k = (\tmop{Tr}_k (i, \tmmathbf{f}) (a_{\leqslant k}))_k . \]
We show that $s (g)$ and $(i, \tmmathbf{f})$ are extentionally
equivalent.  From the definition,
$s (g) = (\tmop{id}_1, \tmmathbf{g})$ where
\[ \tmmathbf{g}_k : A^k \times A \rightarrow A^{k + 1} \times B, \]
\begin{eqnarray*}
  \tmmathbf{g}_k (x, a) & = & ((x, a), (g ((x, a)^{k +}))_k)\\
                        & = & ((x, a), (\tmop{Tr}_k (i, \tmmathbf{f}) ((x, a)^{k +})_{\leqslant
                              k})_k)\\
                        & = & ((x, a), (\tmop{Tr}_k (i, \tmmathbf{f}) (x, a))_k) .
\end{eqnarray*}
We show $(\tmop{id}_1, \tmmathbf{g})$ and $(i, \tmmathbf{f})$ are
extensionally equivalent. For this, we inductively show that for any
$x \in A^{k + 1}$, we have
\[ U (\tmop{id} ; \tmmathbf{g}_0 ; \ldots ; \tmmathbf{g}_k) (x) = (x,
  \tmop{Tr}_k (i, \tmmathbf{f}) (x)) . \] Note that this immediately
entails
$\tmop{Tr}_k (\tmop{id} ; \tmmathbf{g}) = \tmop{Tr}_k (i,
\tmmathbf{f})$. Let $x \in A^k$ and $a_k \in A$.

When $k = 0$, it is obvious.

When $k > 0$, we show
\begin{eqnarray*}
  &  & U (\tmop{id}_1 ; \tmmathbf{g}_0 ; \ldots ; \tmmathbf{g}_k) (x, a_k)\\
  & = & \tmop{let} (s, y) = U (\tmop{id}_1 ; \tmmathbf{g}_0 ; \ldots ;
        \tmmathbf{g}_{k - 1}) (x) \tmop{in}\tmop{let} (s', b) =\tmmathbf{g}_k (s, a_k) \tmop{in} (s', (y, b))\\
  & = & \tmop{let} (s, y) = (x, \tmop{Tr}_{k - 1} (i, \tmmathbf{f}) (x))
        \tmop{in} \tmop{let} (s', b) =\tmmathbf{g}_k (s, a_k) \tmop{in} (s', (y, b))\\
  & = & \tmop{let} y = \tmop{Tr}_{k - 1} (i, \tmmathbf{f}) (x) \tmop{in} \tmop{let} (s', b) =\tmmathbf{g}_k (x, a_k) \tmop{in} (s', (y, b))\\
  & = & \tmop{let} y = \tmop{Tr}_{k - 1} (i, \tmmathbf{f}) (x) \tmop{in} \tmop{let} (s', b) = ((x, a_k), (\tmop{Tr}_k (i, \tmmathbf{f}) (x,
       a_k))_k) \tmop{in} (s', (y, b))\\
  & = & \tmop{let} y = \tmop{Tr}_{k - 1} (i, \tmmathbf{f}) (x) \tmop{in} \tmop{let} b = (\tmop{Tr}_k (i, \tmmathbf{f}) (x, a_k))_k \tmop{in}
       ((x, a_k), (y, b))\\
  & = & ((x, a_k), \tmop{Tr}_k (i, \tmmathbf{f}) (x, a_k)) .
\end{eqnarray*}

\subsection{Proof of \cref{thm:dinaturality}}

Let $r_i = \valtostate{(\id_{Y_i}^h \times g_i)\hc s_i}$.  By
induction, obtain the equality of this:
\begin{equation}\label{eqn:post_dinaturality}
  U(r_0\vc \cdots\vc r_n) \circ (\id_{S_{n+1}\times\prod^n Y_i}\times\sigma_{T_{n+1},U_{n+1}})
\end{equation}
and this:
\begin{equation}\label{eqn:pre_dinaturality}
  U(t_0\vc\cdots\vc t_n) \circ (\id_{\prod^n Y_i}\times U(g_n)\times\id_{T_{n+1}}).
\end{equation}
Then observe the $n$th truncation of $(r_i)$ is
(\ref{eqn:post_dinaturality}) followed by $\pi_{\prod^n Y_i}$ and
similarly, the $n$th truncation of $(t_i)$ is
(\ref{eqn:pre_dinaturality}) followed by the same projection.
Therefore, they are observationally equivalent, though not equal as
sequences.

\subsection{Proofs from section~\ref{sec:seqcdc}}

\subsubsection{Lemma~\ref{lem:leftadd_lifts}}

We already know $\CC$ being Cartesian implies $\caus\CC$ is
  Cartesian.  A commutative monoid structure on $\seq X$ is given by
  componentwise monoid structure:
  $+_{\seq X} \teq \seqof{\cell{+_{\seq X_\bul}}{1}{\seq X_\bul \times
      \seq X_\bul}{1}{\seq X_\bul}}$ and
  $0_{\seq X}\teq \seqof{\cell{0_{\seq X_\bul}}{1}{1}{1}{\seq
      X_\bul}}$.

\subsubsection{Properties of $\squD$}

It is straightforward to check that all the properties we claim are
well-typed, in the sense that the source and target 1-cells of the 2-cells on
each side match. It remains to check that the underlying morphisms for the
2-cells in each claimed property match. This is not an effortless task; we
have two kinds of composition and will be making heavy use of the axioms of
Cartesian differential operators. We find it easiest to do this reasoning with
string diagrams, so this is what we will present here.

When the string diagrams below are particularly complicated, we may draw a
dotted red box around a region or include some red text. Anything found in red
has no mathematical meaning---it is only there to help break down a complex
diagram or foreshadow a major substitution.

Throughout this proof, let $Uf = \phi$, $Ug = \psi$, $Uh = \xi$, and $Uk = \kappa$.

\textbf{Property 1:} If $\phi \in \CC(X, Y)$, then \mbox{$\squD(\phi^h) = (D\phi)^h$}
and \mbox{$\squD(\phi^v) = ((D\phi\times \phi)\circ(X\times \Delta_X))^v$}.

When $\phi$ is oriented horizontally, it has no state. Omitting the wires corresponding to state
in the diagram for the underlying morphism of $\squD$, we find

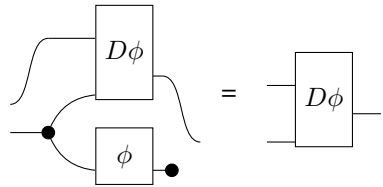
\ctikzfig{squD/prop1hor}

\noindent which matches the underlying morphism of $(D\phi)^h$. On the other hand, when
$\phi$ is oriented vertically, it has no values. Omitting the value wires in the same diagram,
we find

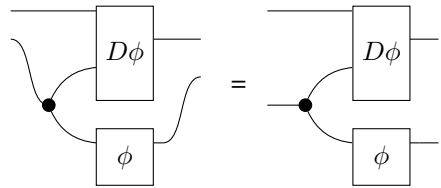
\ctikzfig{squD/prop1vert}

\noindent which again matches.

\textbf{Property 2:} $(0_S\times S)^v\vc \squD f \hc (0_X \times X)^h = (0_Y \circ !_Y)^h \hc f\vc (0_{S'} \times S')^v$

From Lemma~\ref{lem:adjust} and the diagram for the underlying morphism of $\squD f$,

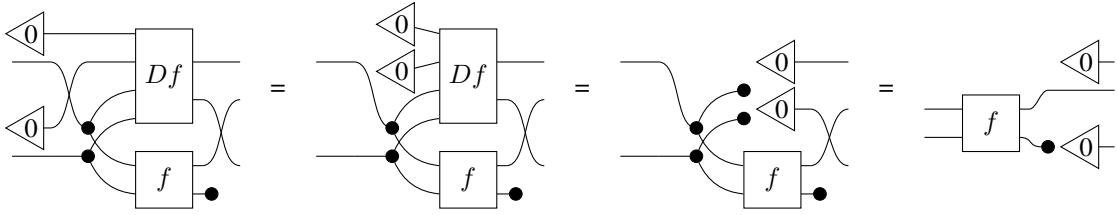
\ctikzfig{squD/prop2}

\noindent where the middle equality is by the CD2 axiom for $D$.

\textbf{Property 3:} $(+_S\times S)^v\vc \squD f \hc (+_X \times X)^h 
      = \alpha_S^v\vc(+_Y^h \hc \squD f \cc \squD f \hc \alpha_X^h)\vc \beta_{S'}^v \vc (+_{S'}\times S')^v$

This property requires the use of CD3 for $D$.

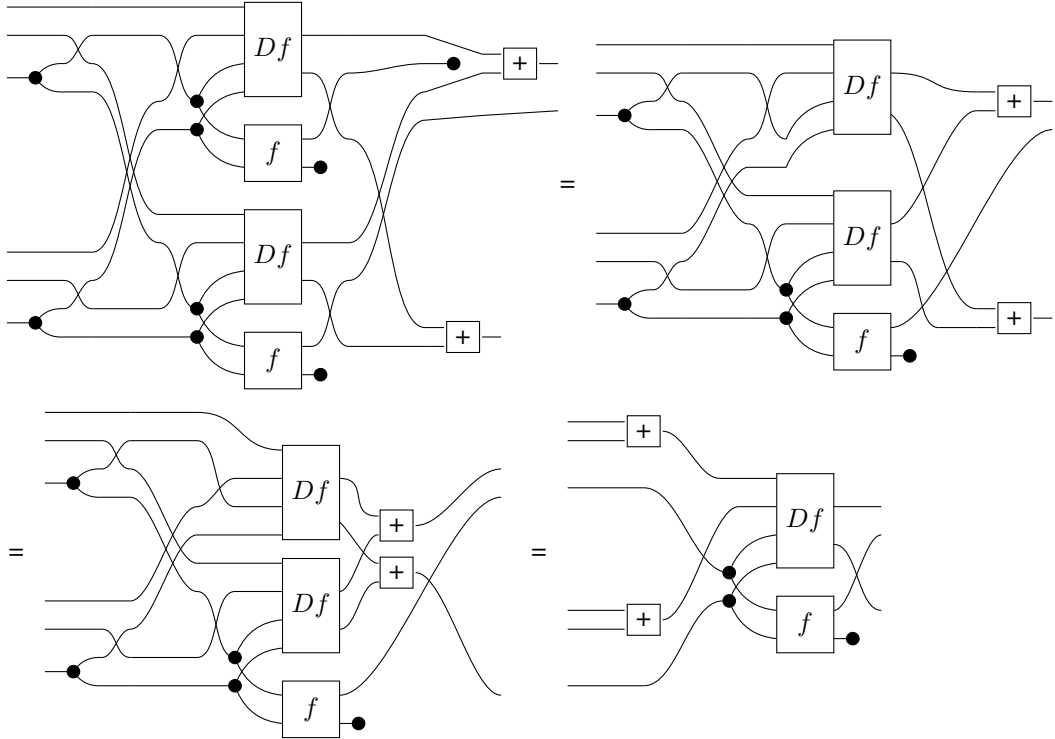
\ctikzfig{squD/prop3partb}


\textbf{Property 4:} $\squD(f\vc h) = (\squD f\vc \squD h) \hc \delta_{X,Z}^h$.

Using the CD axioms for $D$, we compute the differential of the underlying morphism of $f\vc h$:

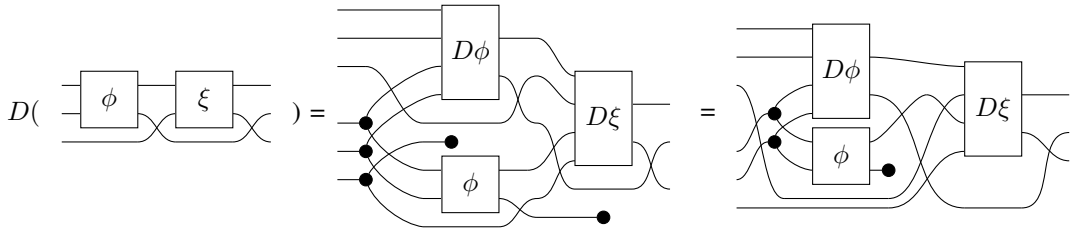
\ctikzfig{squD/prop4preamble}

We use that result to construct a string diagram for $U\squD(f\vc h)$:

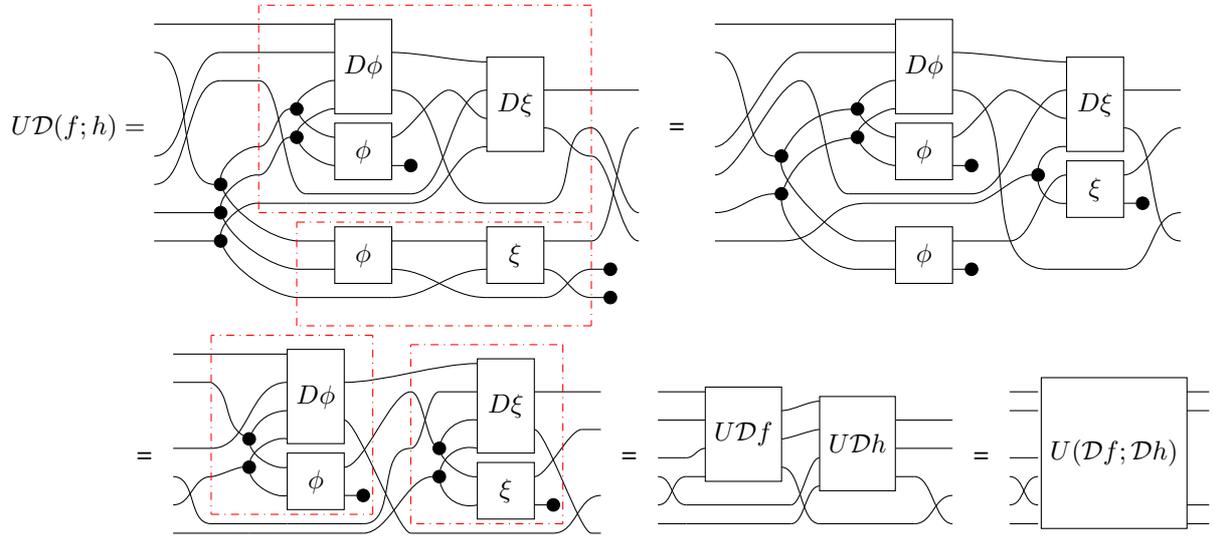
\ctikzfig{squD/prop4}

This is the string diagram of $U((\squD f\vc \squD h) \hc \delta_{X,Z}^h)$ by Lemma~\ref{lem:adjust}.

\textbf{Property 5:} $\squD(g\hc f)\vc \gamma_{S', T'}^v = \gamma_{S, T}^v\vc (\squD g \hc (\squD f \cc f)\hc (X\times \Delta_X)^h)$.

Using the CD axioms for $D$, we compute the differential of the underlying morphism of $g\hc f$:

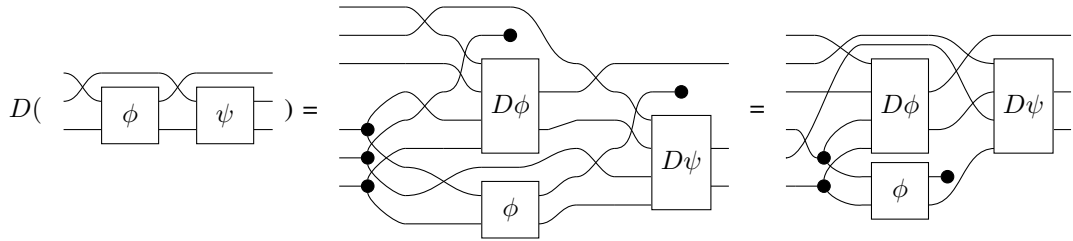
\ctikzfig{squD/prop5preamble}

We use that result to construct a string diagram for $U(\squD(g\hc f)\vc \gamma_{S', T'}^v)$:

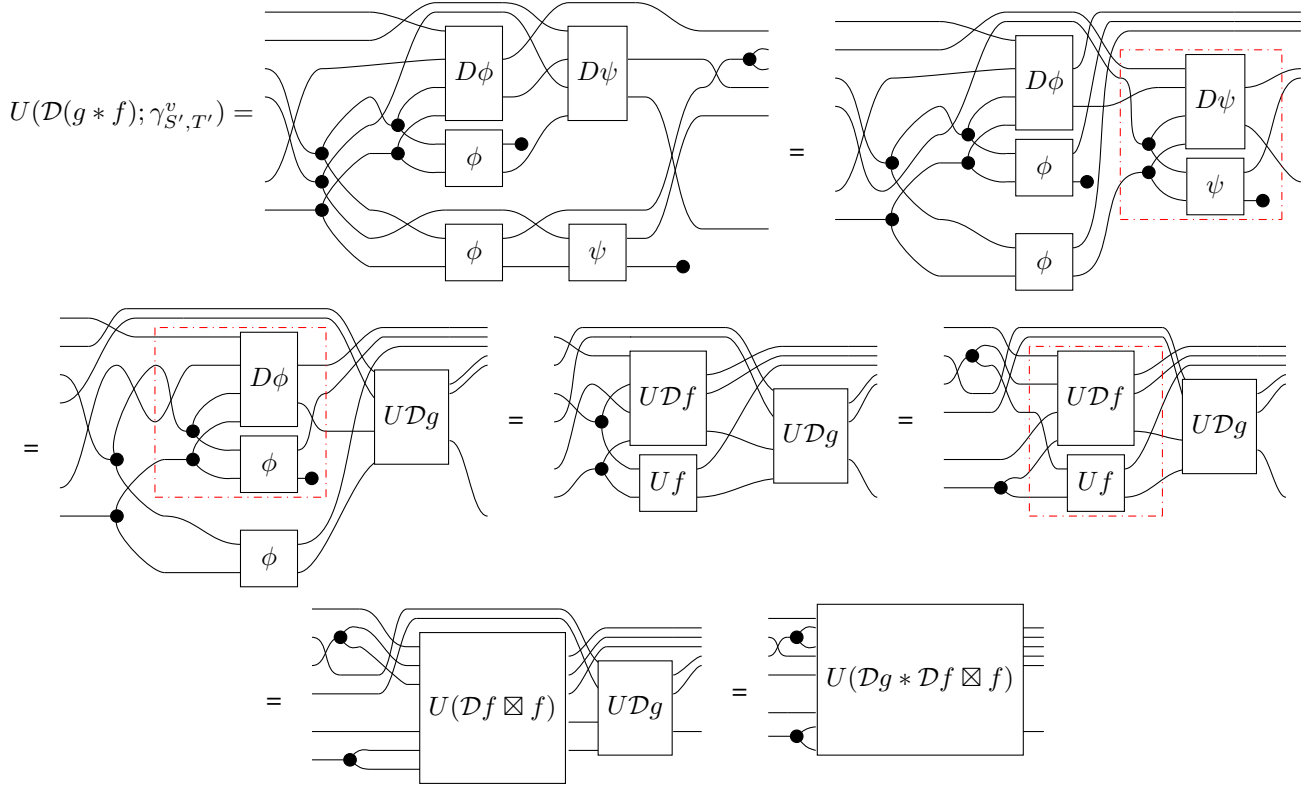
\ctikzfig{squD/prop5row1}
\ctikzfig{squD/prop5row2}
\ctikzfig{squD/prop5row3}

This is the string diagram of $U(\gamma_{S, T}^v\vc (\squD g \hc (\squD f \cc f))\hc (X\times \Delta_X)^h)$ by Lemma~\ref{lem:adjust}.

\textbf{Property 6:} $\squD(f \cc k)\vc \delta_{S', T'}^v = \delta_{S, T}^v\vc \squD f \cc \squD k \hc \delta_{X, Z}^h$.

Using the CD axioms for $D$, we compute the differential of the underlying morphism of $f\cc k$:

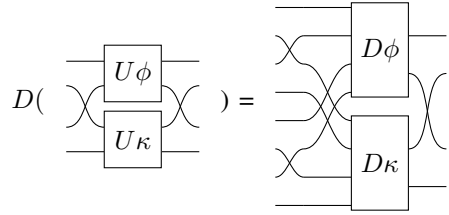
\ctikzfig{squD/prop6preamble}

We use that result to construct a string diagram for $U(\squD(f \times k)\vc \delta_{S', T'}^v)$:

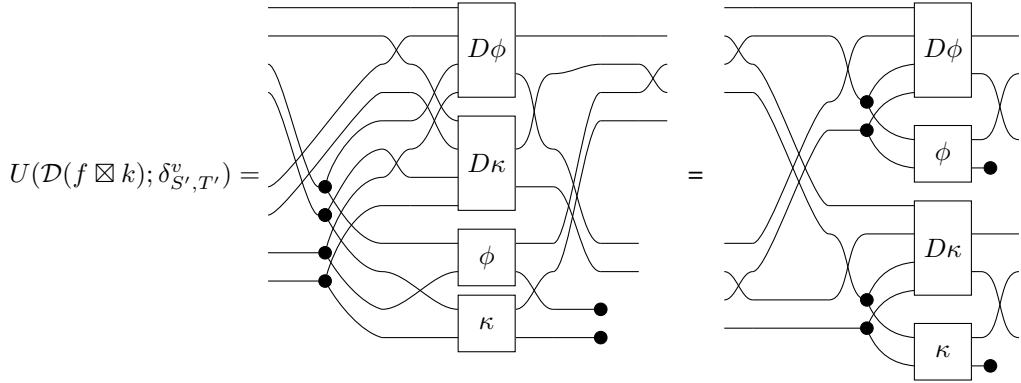
\ctikzfig{squD/prop6}

\textbf{Property 7:} $\zeta_S^v\vc \squD \squD f \hc \zeta_X^h = \squD f\vc \zeta_{S'}^v$.

\textbf{Property 8:} $\delta_{S,S}^v\vc \squD \squD f \hc \delta_{X, X}^h = \squD \squD f\vc \delta_{S', S'}^h$.

For these two properties, we need the underlying morphism for $\squD\squD f$.
This can be found in the diagram below. We do not show all the steps,
hopefully the methodology is clear enough for the impossibly interested reader
to check the claim. For property 7, you need to use properties CD2 and CD6 for
$D$. For property 8, you need to use CD7 for $D$.

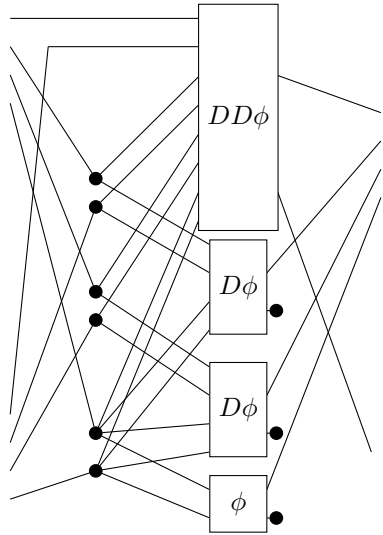
\ctikzfig{squD/doublediff}


\subsubsection{Well-definedness of $\seqD$}

\begin{lemma}
  $\seqD$ is well-defined.
\end{lemma}
\begin{proof}
  Let $(i, \seq s)$ and $(j, \seq t)$ be extensionally equivalent
  sequences.  We must show $(\squD i, \seqof{\squD {\seq s}_\bul})$
  and $(\squD j, \seqof{\squD {\seq t}_\bul})$ are extensionally
  equivalent.
  \begin{eqnarray*}
    \squD i \vc \squD\seq s_0 \vc\ldots\vc \squD\seq s_n \vc d_{(\nxt s_n)^2} 
    &=& \squD i \vc \squD\seq s_0 \vc\ldots\vc \squD\seq s_n \vc \squD d_{\nxt s_n} \\
    &=& \squD(i \vc\seq s_0 \vc\ldots\vc\seq s_n \vc d_{\nxt s_n}) \hc z^{-1} \\
    &=& \squD(j \vc\seq t_0 \vc\ldots\vc\seq t_n \vc d_{\nxt t_n}) \hc z^{-1} \\
  \end{eqnarray*}
  \noindent where the first line is by
  Proposition~\ref{prop:squD_properties}(2), the second by
  Lemma~\ref{lem:squD_vertical_composition}, and the last by the fact
  that $(i, \seq s)$ and $(j, \seq t)$ are extensionally equivalent.
  Here it is crucial that extensionally equivalent sequences have the
  same source values, so that the $z$ obtained from
  Lemma~\ref{lem:squD_vertical_composition} when $\varphi = \seq s$
  matches the $z$ obtained when $\varphi = \seq t$.  The proof then
  finishes by reversing the first two steps.
\end{proof}

\subsubsection{$\seqD$ is a Cartesian differential operator for $\caus\CC$ (Proposition~\ref{pp:cartdiff})}\ \\

We show that $\seqD$ satisfies the seven axioms given in Definition~\ref{def:cdc}. 
Most of the hard work has already been done in 
Proposition~\ref{prop:squD_properties}; we use the properties
of $\squD$ established there in concert with the Shim Lemma to 
obtain most of the CD axioms.

Let $f = (i, \seq f): \seq X \to \seq Y$,
$g = (j, \seq g): \seq Y \to \seq Z$ and
$h = (\ell, \seq h): \seq Z \to \seq W$.  
Let $\seq S$, $\seq T$, and $\seq U$ be the state sequences of $f$, $g$, and $h$, respectively.

\begin{enumerate}[CD1.]
  \item $\seqD s = s\times {\final{\dom (s)}}$ for 
    $s \in \{\seq X, \sigma_{\seq X,\seq Y}, \final{\seq X}, \Delta_{\seq X}, +_{\seq X}, 0_{\seq X}\}$.

    These $s$ are of the form $(\id_1, \seqof{(\seqel s)^h})$ 
    where each $\seqel s$ is a $\CC$-morphism from $\{X, \sigma_{X,Y}, \final{X}, \Delta_X, +_X, 0_X\}$.
    Therefore, 
    \begin{align*}
    \seqD s &= \seqD(\id_1, \seqof{(\seqel s)^h}) = (\id_1, \seqof{\squD(\seqel s)^h}) = (\id_1, \seqof{(D\seqel s)^h}) \\
    & = (\tuple{\id_1,\id_1}, \seqof{(\seqel s\times\final{\dom{\seqel s}})^h})
    = (\id_1, \seqof{(\seqel s)^h}) \times (\id_1, \seqof{(\final{\dom{\seqel s}})^h}) = s \times \final{\dom s}
    \end{align*}

    Where the last step in the first line is by Proposition~\ref{prop:squD_properties}(1),
    and the first step in the second line is by CD1 for $D$.\\

  \item $\seqD (i, \seq f)\circ(0_{\seq X} \times \seq X) = 0_{\seq Y} \circ \final{\seq X}$.

    By definition of $\seqD$, we know $\seqD (i, \seq f)\circ(0_{\seq X} \times \seq X) = (\tuple{0_{\seq S_0}, i}, \seqof{\squD \seqel f \hc (0_{\seqel X}\times\seqel X)^h})$. By Proposition~\ref{prop:squD_properties}(2),
    \begin{displaymath}
      (0_{\seq S_0}\times \seq S_0) \circ i = \tuple{0_{\seq S_0}, i}
      \text{ and }
      (0_{\seqel Y} \circ \final{\seqel Y})^h \hc \seqel f \vc (0_{\nextseqel S} \times \nextseqel S)^v = (0_{\seqel S}\times \seqel S)^v\vc \squD \seqel f \hc (0_{\seqel X} \times {\seqel X})^h,
    \end{displaymath}

    \noindent so the shim lemma tell us

    \begin{align*}
      \seqD (i, \seq f)\circ(0_{\seq X} \times \seq X) 
      &= (\tuple{0_{\seq S_0}, i}, \seqof{\squD \seqel f \hc (0_{\seqel X}\times\seqel X)^h}) & \text{(definitions)}\\
      &= (i, \seqof{(0_{\seqel Y} \circ \final{\seqel Y})^h \hc \seqel f}) & \text{(shim lemma)}\\
      &= 0_{\seq Y} \circ \final{\seq Y} \circ (i, \seq f) & \text{(definitions)}\\
      &= 0_{\seq Y} \circ \final{\seq X} & \text{(finality)}
    \end{align*} 

    \noindent as desired.
    \\

  \item $\seqD (i, \seq f) \circ (+_{\seq X}\times \seq X) = +_{\seq Y}\circ (\seqD (i, \seq f) \times \seqD (i, \seq f)) \circ \alpha_{\seq X}^h$.

    Since $(+_{\seq S_0}\times {\seq S_0})\circ\tuple{0_{\seq S_0}, 0_{\seq S_0}, i} = \tuple{0_{\seq S_0}, i}$,
    and, by Proposition~\ref{prop:squD_properties}(3),

    \begin{displaymath}
      \alpha_{\seqel S}^v\vc (+_{\seqel Y}^h \hc \squD \seqel f \cc \squD \seqel f \hc \alpha_{\seqel X}^h)\vc \beta_{\nextseqel S}^v \vc (+_{\nextseqel S}\times {\nextseqel S})^v = 
      (+_{\seqel S}\times {\seqel S})^v\vc \squD \seqel f \hc (+_{\seqel X} \times {\seqel X})^h,
    \end{displaymath}

    \noindent we can invoke the shim lemma again.
    \begin{align*}
      \seqD (i, \seq f) \circ (+_{\seq X}\times \seq X) &= (\tuple{0_{\seq S_0}, i}, \seqof{\squD \seqel f \hc (+_{\seqel X}\times \seqel X)^h}) & \text{(definitions)}\\
      &= (\tuple{0_{\seq S_0}, 0_{\seq S_0}, i}, 
       [\alpha_{\seqel S}^v\vc (+_{\seqel Y}^h\hc \squD\seqel f \cc \squD\seqel f \hc \alpha_{\seqel X}^h) \vc \beta_{\nextseqel {S}}^v]) & \text{(shim lemma)} \\
      &= (\beta_{\seq S_0}\circ \tuple{0_{{\seq S}_0}, i , 0_{{\seq S}_0}, i}, 
       [\alpha_{\seqel S}^v\vc (+_{\seqel Y}^h\hc \squD\seqel f \cc \squD\seqel f \hc \alpha_{\seqel X}^h) \vc \beta_{\nextseqel {S}}^v]) & \\
      &= (\tuple{0_{{\seq S}_0}, i , 0_{{\seq S}_0}, i}, 
       [\beta_{\seqel {S}}^v\vc\alpha_{\seqel S}^v\vc (+_{\seqel Y}^h\hc \squD\seqel f \times \squD\seqel f \hc \alpha_{\seqel X}^h)]) & \text{(delayed dinaturality)} \\
      &= (\tuple{0_{{\seq S}_0}, i , 0_{{\seq S}_0}, i},
       [(\alpha_{\seqel S}\circ\beta_{\seqel {S}})^v\vc (+_{\seqel Y}^h\hc \squD\seqel f \times \squD\seqel f \hc \alpha_{\seqel X}^h)]) & \\
      &= (\tuple{0_{{\seq S}_0}, i , 0_{{\seq S}_0}, i},
       [+_{\seqel Y}^h\hc \squD\seqel f \times \squD\seqel f \hc \alpha_{\seqel X}^h]) &
    \end{align*}
 
    The last step here is a bit delicate, and involves a special case of the Shim Lemma. 
    If $\seq b$ is a sequence of idempotent maps (i.e.~$\seqel b\circ\seqel b = \seqel b$) such
    that $\seq b_0 \circ i = i$, then $(i, [\seqel f]) = (i, [\seqel b^v; \seqel f])$. Here
    $\alpha_{\seqel S}\circ\beta_{\seqel {S}}$ is such a sequence of idempotent maps.
    \\
  
  \item $\seqD((j, \seq g) \circ (i, \seq f)) = \seqD (j, \seq g) \circ (\seqD (i, \seq f) \times (i, \seq f)) \circ ({\seq X} \times \Delta_{\seq X})$.

    Since $\gamma_{\seq S_0, \seq T_0}\circ\tuple{0_{\seq S_0}, 0_{\seq T_0}, i, j} = 
    \tuple{0_{\seq S_0}, i, i, 0_{\seq T_0}, j}$, and, by Proposition~\ref{prop:squD_properties}(5),
    \[
    \squD(\seqel g\hc \seqel f)\vc \gamma_{\nextseqel S,\nextseqel T}^v = \gamma_{\seqel S,\seqel T}^v\vc (\squD \seqel g \hc (\squD \seqel f \cc \seqel f)\hc (\seqel X\times \Delta_{\seqel X})^h)
    \]
    \noindent we use the shim lemma again.
    \begin{align*}
      \seqD((j, \seq g) \circ (i, \seq f)) 
      &= (\tuple{0_{\seq S_0}, 0_{\seq T_0}, i, j},
      \seqof{\squD(\seqel g\hc \seqel f)}) & \text{(definitions)}\\
      &= (\tuple{0_{\seq S_0}, i, i, 0_{\seq T_0}, j}, 
      \seqof{\squD \seqel g \hc (\squD \seqel f \cc \seqel f)\hc (\seqel X\times \Delta_{\seqel X})^h}) & \text{(shim lemma)} \\
      &= (\tuple{0_{\seq T_0}, j}, \seqof{\squD \seqel g})\circ (\tuple{0_{\seq S_0}, i, i}, \seqof{\squD \seqel f \cc \seqel f})\circ (\seq X\times \Delta_{\seq X}) & \text{(definitions)} \\
      &= \seqD(j, \seq g)\circ (\seqD(i, \seq f) \times (i, \seq f)) \circ (\seq X\times \Delta_{\seq X}) & \text{(definitions)}
    \end{align*}

  \item $\seqD((i, \seq f)\times (\ell ,\seq h)) = (\seqD(i ,\seq f) \times \seqD(\ell ,\seq h))\circ\delta_{\seq X,\seq Z}$.

    Our invocation of the shim lemma this time uses the fact 
    $\delta_{\seq S_0, \seq U_0}\circ \tuple{0_{\seq S_0}, 0_{\seq U_0}, i, \ell} 
    = \tuple{0_{\seq S_0}, i, 0_{\seq U_0}, \ell}$ and Proposition~\ref{prop:squD_properties}(6)
    \[
    \squD(\seqel f \cc \seqel h)\vc \delta_{\nextseqel S, \nextseqel U}^v = \delta_{\seqel S, \seqel U}^v\vc \squD \seqel f \cc \squD \seqel k \hc \delta_{\seqel X, \seqel Z}^h.
    \]
    We can obtain this axiom now as
    \begin{align*}
      \seqD((i, \seq f)\times (\ell, \seq h))
      &= (\tuple{0_{\seq S_0}, i, 0_{\seq U_0}, \ell},
      \seqof{\squD(\seqel f \cc \seqel h)})  &\text{(definitions)}\\
      &= (\tuple{0_{\seq S_0}, 0_{\seq U_0}, i, \ell},
      \seqof{\squD \seqel f \cc \squD \seqel k \hc \delta_{\seqel X, \seqel Z}^h}) &\text{(shim lemma)} \\
      &= (\seqD(i, \seq f)\times\seqD(\ell,\seq h))\circ\delta_{\seqel X, \seqel Z} &\text{(definitions)}
    \end{align*}
  
  \item $\seqD \seqD (i, \seq f)\circ\zeta_{\seq X} = \seqD (i, \seq f)$.

    We invoke the shim lemma this time using the facts that 
    $\zeta_{\seq S_0}\circ \tuple{0_{\seq S_0}, i} = \tuple{0_{\seq S_0}, 0_{\seq S_0}, 0_{\seq S_0}, i}$
    and Proposition~\ref{prop:squD_properties}(7)
    \[
    \squD \seqel f\vc \zeta_{\nextseqel S}^v = \zeta_{\seqel S}^v\vc \squD \squD \seqel f \hc \zeta_{\seqel X}^h.
    \]
    We obtain the axiom directly now,
    \begin{align*}
      \seqD \seqD (i, \seq f)\circ\zeta_{\seq X}
      &= (\tuple{0_{\seq S_0}, 0_{\seq S_0}, 0_{\seq S_0}, i},
      \seqof{\squD \squD \seqel f \hc \zeta_{\seqel X}^h}) & \text{(definitions)}\\
      &= (\tuple{0_{\seq S_0}, i},
      \seqof{\squD \seqel f}) &\text{(shim lemma)}\\
      &= \seqD (i, \seq f)
    \end{align*}
  
  \item $\seqD\seqD (i, \seq f)\circ\delta_{\seq X,\seq X} = \seqD\seqD(i, \seq f)$.

    We invoke the shim lemma with 
    $\delta_{\seq S_0,\seq S_0}\circ \tuple{0_{\seq S_0}, 0_{\seq S_0}, 0_{\seq S_0}, i} = 
    \tuple{0_{\seq S_0}, 0_{\seq S_0}, 0_{\seq S_0}, i}$
    and Proposition~\ref{prop:squD_properties}(8)
    \[
    \squD\squD \seqel f\vc \delta_{\nextseqel S, \nextseqel S}^v = \delta_{\seqel S, \seqel S}^v\vc \squD \squD \seqel f \hc \delta_{\seqel X, \seqel X}^h.
    \]
    We obtain the axiom directly now,
    \begin{align*}
      \seqD \seqD (i, \seq f)\circ\delta_{\seq X,\seq X}
      &= (\tuple{0_{\seq S_0}, 0_{\seq S_0}, 0_{\seq S_0}, i},
      \seqof{\squD \squD \seqel f \hc \delta_{\seqel X,\seqel X}^h}) & \text{(definitions)}\\
      &= (\tuple{0_{\seq S_0}, 0_{\seq S_0}, 0_{\seq S_0}, i},
      \seqof{\squD\squD \seqel f}) &\text{(shim lemma)}\\
      &= \seqD\seqD (i, \seq f)
    \end{align*}

\end{enumerate}

\end{document}